\documentclass[11pt]{article}
\usepackage{fullpage,latexsym,amsmath,amsfonts}

\newcommand{\bitset}{\{0,1\}}
\newcommand{\elll}{\ell}

\newcommand{\GapSVP}{\mbox{\rm GapSVP}}
\newcommand{\GapCVP}{\mbox{\rm GapCVP}}

\DeclareMathOperator{\polylog}{polylog}
\DeclareMathOperator{\poly}{poly}

\newcommand{\eps}{\varepsilon}

\newcommand{\bC}{\mathbb{C}}

\def\01{\{0,1\}}
\newcommand{\ceil}[1]{\lceil{#1}\rceil}
\newcommand{\R}{\mathbb{R}}

\newcommand{\Exp}{\mathbb{E}} %expectation
\newcommand{\ket}[1]{|#1\rangle}
\newcommand{\bra}[1]{\langle#1|}
\newcommand{\ketbra}[2]{|#1\rangle\langle#2|}
\newcommand{\braket}[2]{\langle#1|#2\rangle}
\newcommand{\norm}[1]{{\| #1 \|}}
\newcommand{\inp}[2]{\braket{#1}{#2}} % inproduct, < | >
\newcommand{\Tr}{\mbox{\rm Tr}}

\newcommand{\rank}{\mbox{\rm rank}}
\newcommand{\IP}{\mbox{\rm IP}}
\newcommand{\NP}{\mbox{\rm NP}}
\newcommand{\coNP}{\mbox{\rm coNP}}

\newcommand{\PostBQP}{\mbox{\rm PostBQP}}

\newcommand{\Polytime}{\mbox{\rm P}}

\newcommand{\BQP}{\mbox{\rm BQP}}

\newcommand{\QMA}{\mbox{\rm QMA}}

\newcommand{\PP}{\mbox{\rm PP}}
\newcommand{\UPP}{\mbox{\rm UPP}}

\newtheorem{theorem}{Theorem}[section]
\newtheorem{corollary}[theorem]{Corollary}
\newtheorem{lemma}[theorem]{Lemma}

\newtheorem{fact}[theorem]{Fact}
\newtheorem{claim}[theorem]{Claim}
\newtheorem{definition}[theorem]{Definition}

\newenvironment{proof}
{\noindent {\bf Proof. }}
{{\hfill $\Box$}\\
 \smallskip}

\begin{document}

\title{Quantum Proofs for Classical Theorems}
\author{Andrew Drucker%
\thanks{MIT, adrucker@mit.edu.  Supported by a DARPA YFA grant.} 
\and 
Ronald de Wolf%
\thanks{CWI Amsterdam, rdewolf@cwi.nl. Partially supported by a Vidi grant from the Netherlands Organization for Scientific Research (NWO), and by the European Commission under the Integrated Project Qubit Applications (QAP) funded by the IST directorate as Contract Number 015848.
}}

\date{}
\maketitle

\begin{abstract}
Alongside the development of quantum algorithms and quantum complexity theory in recent years, quantum techniques have also proved instrumental in obtaining results in diverse classical (non-quantum) areas, such as coding theory, 
communication complexity, and polynomial approximations.  
In this paper we survey these results and the quantum toolbox they use.
\end{abstract}

\tableofcontents

\section{Introduction}

\subsection{Surprising proof methods}
Mathematics is full of surprising proofs, and these form a large part of the beauty and fascination of the subject to its practitioners.  A feature of many such proofs is that they introduce objects or concepts from beyond the ``milieu'' in which the problem was originally posed.

As an example from high-school math, the easiest way to prove real-valued identities like 
$$
\cos(x+y)=\cos x\cos y-\sin x\sin y
$$ 
is to go to \emph{complex} numbers: using the identity $e^{ix}=\cos x+i\sin x$ we have
$$
e^{i(x+y)}=e^{ix}e^{iy}=(\cos x+i\sin x)(\cos y+i\sin y)=\cos x\cos y-\sin x\sin y + i(\cos x\sin y+\sin x\cos y)\;.
$$
Taking the real parts of the two sides gives our identity.

Another example is the \emph{probabilistic method}, associated with Paul Erd\H{o}s and excellently covered in the book of Alon and Spencer~\cite{alon&spencer:probmethod}.
The idea here is to prove the existence of an object with a specific desirable property $P$ by choosing such an object at random, 
and showing that it satisfies $P$ with positive probability.
Here is a simple example: suppose we want to prove that every undirected graph $G=(V,E)$ with $|E|=m$ edges has a cut
(a partition $V=V_1\cup V_2$ of its vertex set) with at least $m/2$ edges crossing the cut. 
\begin{quote}
\begin{proof}
Choose the cut at random, by including each vertex $i$ in $V_1$ with probability 1/2 (independently of the other vertices).
For each fixed edge $(i,j)$, the probability that it crosses is the probability that $i$ and $j$ end up in different sets,
which is exactly 1/2. Hence by linearity of expectation, the \emph{expected} number of crossing edges for our cut is exactly $m/2$.  
But then there must exist a specific cut with at least $m/2$ crossing
edges.
\end{proof}
\end{quote}
The statement of the theorem has nothing to do with probability, yet probabilistic methods allow us to give a very simple proof.
Alon and Spencer~\cite{alon&spencer:probmethod} give many other examples of this phenomenon, 
in areas ranging from graph theory and analysis to combinatorics and computer science.

Two special cases of the probabilistic method deserve mention here.
First, one can combine the language of probability with that of information theory~\cite{cover&thomas:infoth}.
For instance, if a random variable $X$ is uniformly distributed over some finite set $S$ then its Shannon entropy 
$H(X)=-\sum_x\Pr[X=x]\log\Pr[X=x]$ is exactly $\log|S|$.
Hence upper (resp.\ lower) bounds on this entropy give upper (resp.\ lower) bounds on the size of~$S$.  
Information theory offers many tools that allow us to manipulate and bound entropies in sophisticated yet intuitive ways.
The following example is due to Peter Frankl. %
In theoretical computer science one often has to bound sums of binomials coefficients like
$$
s=\sum_{i=0}^{\alpha n}\binom{n}{i}\,,
$$
say for some $\alpha\leq 1/2$.
This $s$ is exactly the size of the set $S\subseteq{\bitset}^n$ of strings of Hamming weight at most $\alpha n$.
Choose $X=(X_1,\ldots,X_n)$ uniformly at random from~$S$.  
Then, individually, each $X_i$ is a bit whose probability of being~1 is at most $\alpha$,
and hence $H(X_i)\leq H(\alpha)=-\alpha\log\alpha-(1-\alpha)\log(1-\alpha)$. 
Using the sub-additivity of entropy we obtain an essentially tight upper bound on the size of~$S$:
$$
\log s=\log|S|=H(X)\leq\sum_{i=1}^n H(X_i)\leq nH(\alpha)\,.
$$
A second, related but more algorithmic approach is the so-called ``incompressibility method,'' which reasons about the properties of randomly chosen objects and is based on the theory of Kolmogorov complexity~\cite[Chapter~6]{li&vitanyi:kolm3}.  In this method we consider ``compression schemes,'' that is, injective mappings $C$ from binary strings to other binary strings.  The basic observation is that for any $C$ and $n$, most strings of length $n$ map to strings of length nearly $n$ or more, simply because there aren't enough short descriptions to go round.  Thus, if we can design some compression scheme that represents $n$-bit objects that do \emph{not} have some desirable property $P$ with much fewer than $n$ bits, it follows that most $n$-bit strings have property $P$.

Of course one can argue that applications of the probabilistic method are all just counting arguments disguised in the language of
probability, and hence probabilistic arguments are not essential to the proof. In a narrow sense this is indeed correct.
However, viewing things probabilistically gives a rather different perspective and allows us to use sophisticated
tools to bound probabilities, such as large deviation inequalities and the Lov\'{a}sz Local Lemma, as exemplified in~\cite{alon&spencer:probmethod}.  
While such tools may be viewed as elaborate ways of doing a counting argument, the point is that we might never think 
of using them if the argument were phrased in terms of counting instead of probability.
Similarly, arguments based on information theory or incompressibility are essentially ``just'' counting arguments, 
but the information-theoretic and algorithmic perspective leads to proofs we would not easily discover otherwise.

\subsection{A quantum method?}
The purpose of this paper is to survey another family of surprising proofs
that use the language and techniques of \emph{quantum computing} to prove theorems whose statement has nothing to do with quantum computing.

Since the mid-1990s, especially since Peter Shor's 1994 quantum algorithm for factoring large integers~\cite{shor:factoring},
quantum computing has grown to become a prominent and promising area at the intersection of computer science and physics. 
Quantum computers could yield fundamental improvements in algorithms, communication protocols, and cryptography.
This promise, however, depends on physical realization, and despite the best efforts of experimental physicists we are still very far from building large-scale quantum computers.

In contrast, using the language and tools of quantum computing as a proof tool is something we can do today.
Here, quantum mechanics is purely a mathematical framework, and our proofs remain valid even if large-scale quantum computers are never built (or worse, if quantum mechanics turns out to be wrong as a description of reality).  This paper describes a number of recent results of this type.  As with the probabilistic method, these applications range over many areas, from error-correcting
codes and complexity theory to purely mathematical questions about polynomial
approximations and matrix theory.  We hesitate to say that they represent a ``quantum method,'' since the set of tools is far less developed than the probabilistic method.
However, we feel that these quantum tools will yield more surprises in the future, and have the potential to grow into a full-fledged proof method.

As we will see below, the language of quantum computing is really just a shorthand for linear algebra: 
states are vectors and operations are matrices. Accordingly, one could argue that we don't need the quantum language at all.
Indeed, one can always translate the proofs given below back to the language of linear algebra.  
What's more, there is already an extensive tradition of elegant proofs in combinatorics, geometry, and other areas, 
which employ linear algebra (often over finite fields) in surprising ways.  
For two surveys of this \emph{linear algebra method}, see the books
by Babai and Frank~\cite{babai&frankl:linalg} and Jukna~\cite[Part~III]{jukna:excom}.
However, we feel the proofs we survey here are of a different nature than those produced by the classical linear algebra method.
Just as thinking probabilistically suggests strategies that might not occur when taking the counting perspective, couching a problem in the language of quantum algorithms and quantum information gives us access to intuitions and tools that we would otherwise likely overlook or consider unnatural.
While certainly not a cure-all, for some types of problems the quantum perspective is a very useful one
and there is no reason to restrict oneself to the language of linear algebra.
\subsection{Outline}
The survey is organized as follows.
We begin in Section~\ref{sectoolbox} with a succinct introduction to the quantum model and the properties used 
in our applications.  Most of those applications can be conveniently classified in two broad categories.
First, there are applications that are close in spirit to the classical information-theory method.  They use quantum information
theory to bound the dimension of a quantum system, analogously to how classical information theory
can be used to bound the size of a set. In Section~\ref{secappqit} we give three results of this type.
Other applications use quantum algorithms as a tool to define \emph{polynomials} with desired properties.
In Section~\ref{secapppoly} we give a number of applications of this type.
Finally, there are a number of applications of quantum tools that do not fit well in the previous two categories;
some of these are classical results more indirectly ``inspired'' by earlier quantum results.
These are described in Section~\ref{secotherapps}.

\section{The quantum toolbox}\label{sectoolbox}

The goal of this survey is to show how quantum techniques can be used to analyze non-quantum questions.
Of course, this requires at least \emph{some} knowledge of quantum mechanics, which might appear discouraging to those without a physics background.
However, the amount of quantum mechanics one needs is surprisingly small and easily explained in terms of basic linear algebra.  The first thing we would like to convey is that at the basic level, quantum mechanics is not a full-fledged theory of the universe (containing claims about which objects and forces ``really exist''), but rather a \emph{framework} in which to describe physical systems and processes they undergo.  Within this framework we can posit the existence of basic units of quantum information (``qubits'') and ways of transforming them, just as classical theoretical computer science begins by positing the existence of bits and the ability to perform basic logical operations on them.  While we hope this is reassuring, it is nevertheless true that the quantum-mechanical framework has strange and novel aspects---which, of course, is what makes it worth studying in the first place.

In this section we give a bare-bones introduction to the essentials of quantum mechanics and quantum computing.
(A more general framework for quantum mechanics is given in Appendix~\ref{appgeneral}, but we will not need it for the results we describe.)  We then give some specific useful results from quantum information theory and quantum algorithms.

\subsection{The quantum model}\label{secquantummodel}

At a very general level, any physical system is associated with a \emph{Hilbert space},
and a state of that system is described by an element of that Hilbert space.
The Hilbert space corresponding to the combination of two physical systems is the \emph{tensor} product of their respective Hilbert spaces.

\paragraph{Pure states} 
For our purposes, a \emph{pure} quantum state (often just called a \emph{state}) is a unit column vector in 
a $d$-dimensional complex vector space $\bC^d$. Quantum physics typically used the \emph{Dirac notation}, 
writing a column vector $v$ as $\ket{v}$, while $\bra{v}$ denotes the row vector that is the conjugate transpose of~$v$.

The simplest nontrivial example is the case of a 2-dimensional system, called a \emph{qubit}.
We identify the two possible values of a classical bit with the two vectors in the standard orthonormal basis for this space:
$$
\ket{0}=\left(\begin{array}{r}1\\ 0\end{array}\right), \quad
\ket{1}=\left(\begin{array}{r}0\\ 1\end{array}\right).
$$
In general, the state of a qubit can be a \emph{superposition} (i.e., linear combination) of these two values:
$$
\ket{\phi} = \alpha_0\ket{0}+\alpha_1\ket{1}=\left(\begin{array}{r}\alpha_0 \\ \alpha_1\end{array}\right)\,,
$$
where the complex numbers are called \emph{amplitudes}; $\alpha_0$ is the amplitude of basis state $\ket{0}$, and $\alpha_1$ is the amplitude of $\ket{1}$.
Since a state is a unit vector, we have $|\alpha_0|^2+|\alpha_1|^2=1$.

A 2-qubit space is obtained by taking the \emph{tensor product} of two 1-qubit spaces.
This is most easily explained by giving the four basis vectors of the tensor space:
\begin{align*}
\ket{00} &=\ket{0}\otimes\ket{0}=\left(\begin{array}{r}1\\ 0\\ 0\\ 0\end{array}\right),
&\ket{01}&=\ket{0}\otimes\ket{1}=\left(\begin{array}{r}0\\ 1\\ 0\\ 0\end{array}\right),\\
\ket{10}&=\ket{1}\otimes\ket{0}=\left(\begin{array}{r}0\\ 0\\ 1\\ 0\end{array}\right),
&\ket{01}&=\ket{1}\otimes\ket{1}=\left(\begin{array}{r}0\\ 0\\ 0\\ 1\end{array}\right).
\end{align*}
These correspond to the four possible 2-bit strings.  
More generally, we can form $2^n$-dimensional spaces this way whose basis states correspond to the $2^n$ different $n$-bit strings.

We also sometimes use $d$-dimensional spaces without such a qubit-structure.
Here we usually denote the $d$ standard orthonormal basis vectors with $\ket{1},\ldots,\ket{d}$, where $\ket{i}_i=1$ and $\ket{i}_j=0$ for all $j\neq i$.
For a vector $\ket{\phi}=\sum_{i=1}^d\alpha_i\ket{i}$ in this space, $\bra{\phi}=\sum_{i=1}^d\alpha_i^*\bra{i}$ 
is the row vector that is the conjugate transpose of $\ket{\phi}$.  The Dirac notation allows us for instance 
to conveniently write the standard inner product between states $\ket{\phi}$ and $\ket{\psi}$ as $\bra{\phi}\cdot\ket{\psi}=\inp{\phi}{\psi}$.
This inner product induces the Euclidean norm (or ``length'') of vectors: $\norm{v}=\sqrt{\inp{v}{v}}$.

One can also take tensor products in this space:  
if $\ket{\phi} = \sum_{i \in [m]} \alpha_i \ket{i}$ and $\ket{\psi} = \sum_{j \in [n]} \beta_j \ket{j}$, 
then their tensor product $\ket{\phi} \otimes \ket{\psi} \in \mathbb{C}^{mn}$ is 
$$
\ket{\phi} \otimes \ket{\psi} = \sum_{i\in [m], j\in [n]} \alpha_i \beta_j \ket{i, j}\,,
$$
where $[n]$ denotes the set $\{1,\ldots,n\}$ and 
the vectors $\ket{i, j}=\ket{i}\otimes\ket{j}$ form an orthonormal basis for  $\mathbb{C}^{mn}$.  
This tensor product of states $\ket{\phi}$ and $\ket{\psi}$ is also often denoted simply as $\ket{\phi}\ket{\psi}$.  
Note that this new state is a unit vector, as it should be.  

Not every pure state in $\mathbb{C}^{mn}$ can 
be expressed as a tensor product in this way; those that cannot are called \emph{entangled}.  
The best-known entangled state is the 2-qubit \emph{EPR-pair}
$({1}/{\sqrt{2}}) \left(\ket{00}+\ket{11}\right)$, named after the authors of the paper~\cite{epr}.
When two separated parties each hold part of such an entangled state, we talk about \emph{shared} entanglement between them.

\paragraph{Transformations} 
There are two things one can do with a quantum state: transform it or measure it.  Actually, as we will see, measurements can transform the measured states as well; 
however, we reserve the word ``transformation'' to describe non-measurement change processes, which we describe next.  
Quantum mechanics allows only \emph{linear} transformations on states.
Since these linear transformations must map unit vectors to unit vectors, 
we require them to be \emph{norm-preserving} (equivalently, \emph{inner-product-preserving}).
Norm-preserving linear maps are called \emph{unitary}.
Equivalently, these are the $d\times d$ matrices $U$ whose conjugate transpose $U^*$ equals the inverse $U^{-1}$
(physicists typically write $U^\dagger$ instead of $U^*$).
For our purposes, unitary transformations are exactly the transformations that quantum mechanics allows us to apply to states.  We will frequently define transformations by giving their action on the standard basis, with the understanding that such a definition extends (uniquely) to a linear map on the entire space.

Possibly the most important 1-qubit unitary is the \emph{Hadamard transform}:
\begin{equation}\label{eqhadamard}
\frac{1}{\sqrt{2}}\left(\begin{array}{rr}1 & 1\\ 1 & -1\end{array} \right)\,.
\end{equation}
This maps basis state $\ket{0}$ to $\tfrac{1}{\sqrt{2}}(\ket{0}+\ket{1})$ and $\ket{1}$ to $\tfrac{1}{\sqrt{2}}(\ket{0}-\ket{1})$.

Two other types of unitaries deserve special mention.  First, for any
function $f: {\bitset}^n \rightarrow {\bitset}^n$, define a
transformation $U_f$ mapping the joint computational basis state
$\ket{x}\ket{y}$ (where $x, y \in {\bitset}^n$) to
$\ket{x}\ket{y\oplus f(x)}$, where ``$\oplus$'' denotes bitwise addition ($\bmod 2$) of $n$-bit vectors.  Note that $U_f$ is a permutation on the orthonormal basis states, and therefore unitary.  With such transformations we can simulate classical computations.
Next, fix a unitary transformation $U$ on a $k$-qubit system, and consider the $(k+1)$-qubit unitary transformation $V$ defined by
\begin{equation}
V (\ket{0}\ket{\psi}) = \ket{0}\ket{\psi},\quad  V (\ket{1}\ket{\psi}) = \ket{1}U\ket{\psi}\,.
\end{equation}
This $V$ is called a \emph{controlled-$U$} operation, and the first qubit is called the \emph{control qubit}.  
Intuitively, our quantum computer uses the first qubit to ``decide'' whether or not to apply $U$ to the last $k$ qubits.  

Finally, just as one can take the tensor product of quantum states in different registers, one can take the tensor product of quantum operations (more generally, of matrices) acting on two registers. If $A=(a_{ij})$ is an $m\times m'$ matrix and $B$ is an $n\times n'$ matrix, then their tensor product is the $mn\times m'n'$ matrix
$$
A\otimes B=\left(\begin{array}{rrr}a_{11}B & \cdots & a_{1m'}B\\
                            a_{21}B & \cdots & a_{2m'}B\\
                                    & \ddots & \\
                            a_{m1}B & \cdots & a_{mm'}B
\end{array}
\right).
$$
Note that the tensor product of two vectors is the special case where $m'=n'=1$, and that $A\otimes B$ is unitary if $A$ and $B$ are.  We may regard $A\otimes B$ as the simultaneous application of $A$ to the first register and $B$ to the second register.
For example, the $n$-fold tensor product $H^{\otimes n}$ denotes the unitary that applies the one-qubit Hadamard gate to each qubit of an $n$-qubit register.
This maps any basis state $\ket{x}$ to 
$$
H^{\otimes
  n}\ket{x}=\frac{1}{\sqrt{2^n}}\sum_{y\in{\bitset}^n}(-1)^{x\cdot
  y}\ket{y}
$$
(and vice versa, since $H$ happens to be its own inverse). 
Here $x\cdot y=\sum_{i=1}^n x_iy_i$ denotes the inner product of bit strings. 

\paragraph{Measurement}
Quantum mechanics is distinctive for having \emph{measurement} built-in as a fundamental notion, at least in most formulations.  
A measurement is a way to obtain information about the measured quantum system.  It takes as input a quantum state
and outputs classical data (the ``measurement outcome''), along with a new quantum state. 
It is an inherently probabilistic process that affects the state being measured.
Various types of measurements on systems are possible. 
In the simplest kind, known as \emph{measurement in the computational basis}, we measure a pure state 
$$
\ket{\phi}=\sum_{i=1}^d\alpha_i\ket{i}
$$
and see the basis state $\ket{i}$ with probability $p_i$ equal to the squared amplitude $|\alpha_i|^2$ (or more accurately, the squared \emph{modulus} of the amplitude---it is often convenient to just call this the squared amplitude).  Since the state is a unit vector these outcome probabilities sum to~1, as they should.  After the measurement, the state has changed to the observed basis state $\ket{i}$.
Note that if we apply the measurement now a second time, we will observe the same $\ket{i}$ with certainty---as if the first measurement 
forced the quantum state to ``make up its mind.''

A more general type of measurement is the \emph{projective}
measurement, also known as \emph{Von Neumann measurement}.
A projective measurement with $k$ outcomes is specified by $d\times d$ \emph{projector matrices} $P_1,\ldots,P_k$ that form an orthogonal
decomposition of the $d$-dimensional space. That is, $P_iP_j=\delta_{i,j}P_i$, and $\sum_{i=1}^k P_i = I$ is the identity operator on the whole space.
Equivalently, there exist orthonormal vectors $v_1,\ldots,v_d$ and a partition $S_1\cup\cdots\cup S_k$ of $\{1,\ldots,d\}$
such that $P_i=\sum_{j\in S_i}\ketbra{v_j}{v_j}$ for all $i\in[k]$.
With some abuse of notation we can identity $P_i$ with the subspace onto which it projects, 
and write the orthogonal decomposition of the complete space as
$$
\mathbb{C}^d=P_1\oplus P_2\oplus\cdots\oplus P_k\,.
$$
Correspondingly, we can write $\ket{\phi}$ as the sum of its components in the $k$ subspaces:
$$
\ket{\phi}=P_1\ket{\phi}+P_2\ket{\phi}+\cdots+P_k\ket{\phi}\,.
$$ 
A measurement probabilistically picks out one of these components:
the probability of outcome $i$ is $\norm{P_i\ket{\phi}}^2$, and if we got outcome $i$ then
the state changes to the new unit vector $P_i\ket{\phi}/\norm{P_i\ket{\phi}}$
(which is the component of $\ket{\phi}$ in the $i$-th subspace, renormalized). 

An important special case of projective measurements is \emph{measurement relative to the orthonormal basis $\{\ket{v_i}\}$}, where each projector $P_i$ projects onto a 1-dimensional subspace spanned by the unit vector $\ket{v_i}$.  In this case we have $k=d$ and $P_i = \ketbra{v_i}{v_i}$.  A measurement in the computational basis corresponds to the case where $P_i=\ketbra{i}{i}$.  If $\ket{\phi}=\sum_i\alpha_i\ket{i}$ then we indeed recover the squared amplitude:
$p_i=\norm{P_i\ket{\phi}}^2=|\alpha_i|^2$.

One can also apply a measurement to part of a state, for instance to the first register of a 2-register quantum system.
Formally, we just specify a $k$-outcome projective measurement for the first register, 
and then tensor each of the $k$ projectors with the identity operator on the second register 
to obtain a $k$-outcome measurement on the joint space.

Looking back at our definitions, we observe that if two quantum states $\ket{\phi}, \ket{\psi}$ satisfy $\alpha \ket{\phi}= \ket{\psi}$ for some scalar $\alpha$ (necessarily of unit norm), then for any system of projectors $\{P_i\}$, $\norm{P_i\ket{\phi}}^2 = \norm{P_i\ket{\psi}}^2$ and so measuring $\ket{\phi}$ with $\{P_i\}$ yields the same distribution as measuring $\ket{\psi}$.  More is true: if we make any sequence of transformations and measurements to the two states, the sequence of measurement outcomes we see are identically distributed.  Thus the two states are indistinguishable, and we generally regard them as \emph{the same state}.

\paragraph{Quantum-classical analogy}
For the uninitiated, these high-dimensional complex vectors and unitary transformations may seem baffling.
One helpful point of view is the analogy with classical random processes.
In the classical world, the evolution of a probabilistic automaton whose state consists of $n$ bits can be modeled as a sequence of $2^n$-dimensional vectors $\pi^1, \pi^2, \ldots$.
Each $\pi^i$ is a probability distribution on ${\bitset}^n$, where $\pi^t_x$ gives the probability that the automaton is in state $x$ if measured at time $t$ ($\pi^1$ is the starting state).  
The evolution from time $t$ to $t+1$ is describable by a matrix equation $\pi^{t+1} = M_t\pi^t$, where $M_t$ is a $2^n\times 2^n$ stochastic matrix, that is, a matrix that always maps probability vectors to probability vectors.  
The final outcome of the computation is obtained by sampling from the last probability distribution.
The quantum case is similar: an $n$-qubit state is a $2^n$-dimensional vector, but now it is a vector of complex numbers
whose \emph{squares} sum to~1. A transformation corresponds to a $2^n\times 2^n$ matrix, 
but now it is a matrix that preserves the sum of squares of the entries.
Finally, a measurement in the computational basis obtains the final outcome by sampling from the distribution given by the squares of the entries of the vector.

\subsection{Quantum information and its limitations}\label{ssecqinfo}

Quantum information theory studies the quantum generalizations of familiar notions from classical information theory
such as Shannon entropy, mutual information, channel capacities, etc.
In Section~\ref{secappqit} we give several examples where quantum information theory is used to say something
about various non-quantum systems. The quantum information-theoretic results we need all have the same flavor:
they say that a low-dimensional quantum state (i.e., a small number of qubits) cannot contain too much \emph{accessible} information.  

\paragraph{Holevo's Theorem}
The mother of all such results is Holevo's theorem from 1973~\cite{holevo}, which predates the area of quantum computation by many years.
Its proper technical statement is in terms of a quantum generalization of mutual information, 
but the following consequence of it (derived by Cleve et al.~\cite{cdnt:ip}) about two communicating parties, 
suffices for our purposes.

\begin{theorem}[Holevo, CDNT]\label{thholevo}
If Alice wants to send $n$ bits of information to Bob via a quantum channel (i.e., by exchanging quantum systems),
and they do not share an entangled state, then they have to exchange at least $n$ qubits.
If they are allowed to share unlimited prior entanglement, then Alice has to send
at least $n/2$ qubits to Bob, no matter how many qubits Bob sends to Alice.
\end{theorem}

This theorem is slightly imprecisely stated here, but the intuition is very clear: the first part of the theorem says that if we encode some classical
random variable $X$ in an $m$-qubit state,\footnote{Via an encoding map $x\mapsto\ket{\phi_x}$; we generally use capital letters 
like $X$ to denote random variables, lower case like $x$ to denote specific values.} then no measurement on the quantum
state can give more than $m$ bits of information about $X$. More precisely: the classical mutual information between $X$ 
and the classical measurement outcome $M$ on the $m$-qubit system, is at most $m$.
If we encoded the classical information in a $m$-\emph{bit} system instead of a $m$-qubit system
this would be a trivial statement, but the proof of Holevo's theorem is quite non-trivial.
Thus we see that a $m$-qubit state, despite somehow ``containing'' $2^m$ complex amplitudes,
is no better than $m$ classical bits for the purpose of storing information (this is in the absence of prior entanglement;
if Alice and Bob do share entanglement, then $m$ qubits are no better than $2m$ classical bits).

\paragraph{Low-dimensional encodings}
Here we provide a ``poor man's version'' of Holevo's theorem due to Nayak~\cite[Theorem~2.4.2]{nayak:qfa}, 
which has a simple proof and often suffices for applications.
Suppose we have a classical random variable $X$, uniformly distributed over $[N]=\{1,\ldots,N\}$.
Let $x\mapsto\ket{\phi_x}$ be some encoding of $[N]$, where $\ket{\phi_x}$ is a pure state in a $d$-dimensional space.
Let $P_1,\ldots,P_N$ be the measurement operators applied for decoding; these sum to the $d$-dimensional identity operator.
Then the probability of correct decoding in case $X=x$, is 
$$
p_x=\norm{P_x\ket{\phi_x}}^2\leq \Tr(P_x)\,.
$$
The sum of these success probabilities is at most
\begin{equation}\label{nayakpibound}
\sum_{x=1}^N p_x\leq \sum_{x=1}^N \Tr(P_x)=\Tr\left(\sum_{x=1}^N P_x\right)=\Tr(I)=d\,.
\end{equation}
In other words, if we are encoding one of $N$ classical values in a $d$-dimensional quantum state, 
then any measurement to decode the encoded classical value has average success probability at most $d/N$ 
(uniformly averaged over all $N$ values that we can encode).%
\footnote{For projective measurements the statement is somewhat trivial, since in a $d$-dimensional space 
one can have at most $d$ non-zero orthogonal projectors. However, the same proof works for 
the most general states and measurements that quantum mechanics allows: so-called mixed states 
(probability distributions over pure states) and POVMs 
(which are measurements where the operators $P_1,\ldots,P_k$ need not be projectors, but can be general positive semidefinite matrices summing to $I$);
see Appendix~\ref{appgeneral} for these notions.}
This is optimal.
For example, if we encode $n$ bits into $m$ qubits, we will have $N=2^n$, $d=2^m$, 
and the average success probability of decoding is at most $2^m/2^n$.

\paragraph{Random access codes}
The previous two results dealt with the situation where we encoded a classical random variable $X$ in
some quantum system, and would like to recover the original value $X$ by an appropriate
measurement on that quantum system.  However, suppose $X=X_1\ldots X_n$ is a string of $n$ bits,
uniformly distributed and encoded by a map $x\mapsto\ket{\phi_x}$, and it suffices for us if we are able to decode individual bits $X_i$ from this with some probability $p>1/2$.  More precisely, for each $i\in[n]$ there should exist 
a measurement $\{M_i,I-M_i\}$ allowing us to recover $x_i$: for each $x\in{\bitset}^n$ 
we should have $\norm{M_i\ket{\phi_x}}^2\geq p$ if $x_i=1$ and $\norm{M_i\ket{\phi_x}}^2\leq 1-p$ if $x_i=0$.
An encoding satisfying this is called a \emph{quantum random access code}, since it allows
us to choose which bit of $X$ we would like to access. Note that the measurement to recover $x_i$ 
can change the state $\ket{\phi_x}$, so generally we may not be able to decode more than one bit of $x$.

An encoding that allows us to recover an $n$-bit string requires about $n$ qubits by Holevo.
Random access codes only allow us to recover \emph{each} of the $n$ bits.
Can they be much shorter? In small cases they can be: for instance, one can encode two classical bits
into one qubit, in such a way that each of the two bits can be recovered with 
success probability $85\%$ from that qubit~\cite{ambainis:racj}.
However, Nayak~\cite{nayak:qfa}
proved that asymptotically quantum random access codes cannot be much shorter than classical
(improving upon an $m=\Omega(n/\log n)$ lower bound from~\cite{ambainis:racj}).

\begin{theorem}[Nayak]\label{thrac}
Let $x\mapsto\ket{\phi_x}$ be a quantum random access encoding of $n$-bit strings into $m$-qubit states such that, for each $i \in [n]$, we can decode $X_i$ from $\ket{\phi_X}$ with success probability $p$ (over a uniform choice of $X$ and the measurement randomness).
Then $m\geq (1-H(p))n$, where $H(p)=-p\log p - (1-p)\log(1-p)$ is the binary entropy function. 
\end{theorem}

In fact the success probabilities need not be the same for all $X_i$; if we can decode each $X_i$ with success probability $p_i\geq 1/2$, 
then the lower bound on the number of qubits is $m\geq \sum_{i=1}^n(1-H(p_i))$.
The intuition of the proof is quite simple: since the quantum state allows us to predict the bit $X_i$ with
probability $p_i$, it reduces the ``uncertainty'' about $X_i$ from 1~bit to $H(p_i)$~bits.
Hence it contains at least $1-H(p_i)$ bits of information about $X_i$. Since all $n$ $X_i$'s
are independent, the state has to contain at least $\sum_{i=1}^n(1-H(p_i))$ bits about $X$ in total.
For more technical details see~\cite{nayak:qfa} or Appendix~B of~\cite{kerenidis&wolf:qldcj}.
The lower bound on $m$ can be achieved up to an additive $O(\log n)$ term, even by classical probabilistic encodings.

\subsection{Quantum query algorithms}\label{ssecqueryalgos}

Different models for quantum algorithms exist.
Most relevant for our purposes are the \emph{quantum query algorithms},
which may be viewed as the quantum version of classical decision trees.
We will give a basic introduction here, referring to~\cite{buhrman&wolf:dectreesurvey} for more details.
The model and results of this section will not be needed until Section~\ref{secapppoly}, 
and the reader might want to defer reading this until they get there.

\paragraph{The query model}
In this model, the goal is to compute some function $f:A^n\rightarrow B$ on a given input $x\in A^n$.
The simplest and most common case is $A=B={\bitset}$.
The distinguishing feature of the query model is the way $x$ is accessed:
$x$ is not given explicitly, but is stored in a random access memory, and we are being charged unit
cost for each \emph{query} that we make to this memory.  
Informally, a query asks for and receives the $i$-th element $x_i$ of the input.
Formally, we model a query unitarily as the following 2-register quantum operation $O_x$, 
where the first register is $n$-dimensional and the second is $|A|$-dimensional:
$$
O_x:\ket{i,b}\mapsto\ket{i,b+x_i}\,,
$$
where for simplicity we identify $A$ with the additive group $\mathbb{Z}_{|A|}$, i.e., addition is modulo $|A|$.
In particular, $\ket{i,0}\mapsto\ket{i,x_i}$.
This only states what $O_x$ does on basis states, but by linearity determines the full unitary.
Note that a quantum algorithm can apply $O_x$ to a superposition of basis states; this gives us a kind of simultaneous access to multiple input variables $x_i$.

A $T$-query quantum algorithm starts in a fixed state, say the all-0 state $\ket{0\ldots 0}$, and then interleaves
fixed unitary transformations $U_0,U_1,\ldots,U_T$ with queries.
It is possible that the algorithm's fixed unitaries act on a workspace-register, in addition to the two registers on which $O_x$ acts.
In this case we implicitly extend $O_x$ by tensoring it with the identity operation on this extra register.
Hence the final state of the algorithm can be written as the following matrix-vector product:
$$
U_TO_xU_{T-1}O_x\cdots O_xU_1O_xU_0\ket{0\ldots 0}\,.
$$
This state depends on the input $x$ only via the $T$ queries.
The output of the algorithm is obtained by a measurement of the final state.
For instance, if the output is Boolean, the algorithm could just measure the final state 
in the computational basis and output the first bit of the result.

The \emph{query complexity} of some function $f$ is now defined to be the minimal number of queries needed 
for an algorithm that outputs the correct value $f(x)$ for every $x$ in the domain of $f$
(with error probability at most some fixed value $\eps$).
We just count queries to measure the complexity of the algorithm, while the intermediate fixed unitaries are treated as costless.
In many cases, including all the ones in this paper, the overall computation time of quantum query algorithms 
(as measured by the total number of elementary gates, say) is not much bigger than the query complexity.  
This justifies analyzing the latter as a proxy for the former.

\paragraph{Examples of quantum query algorithms}
Here we list a number of quantum query algorithms that we will need in later sections.  All of these algorithms outperform the best classical algorithms for the given task.
\begin{itemize}
\item {\bf Grover's algorithm}~\cite{grover:search} searches for a ``solution'' in a given $n$-bit input $x$, i.e., an index $i$ such that $x_i=1$.
The algorithm uses $O(\sqrt{n})$ queries, and if there is at least one solution in $x$ then it finds one with probability at least 1/2.  Classical algorithms for this task, including randomized ones, require $\Omega(n)$ queries.
\item {\bf $\eps$-error search:} 
If we want to reduce the error probability in Grover's search algorithm to some small $\eps$,
then $\Theta(\sqrt{n\log(1/\eps)})$ queries are necessary and 
sufficient~\cite{bcwz:qerror}.
Note that this is more efficient than the standard amplification that repeats Grover's algorithm $O(\log(1/\eps))$ times.
\item {\bf Exact search:} 
If we know there are exactly $t$ solutions in our space (i.e., $|x|=t$), then a variant of Grover's algorithm
finds a solution with probability~1 using $O(\sqrt{n/t})$ queries~\cite{bhmt:countingj}.
\item {\bf Finding all solutions:}
If we know an upper bound $t$ on the number of solutions (i.e., $|x|\leq t$),
then we can find all of them with probability~1 using $\sum_{i=1}^t O(\sqrt{n/i})=O(\sqrt{tn})$ queries~\cite{graaf&wolf:qyao}.
\item {\bf Quantum counting:}
The algorithm $Count(x, T)$ of~\cite{bhmt:countingj} approximates the total number of solutions.
It takes as input an $x \in {\bitset}^n$, makes $T$ quantum queries to $x$, and outputs an estimate $\tilde{t} \in [0, n]$ to $t=|x|$, the Hamming weight of $x$.  For $j \geq 1$ we have the following concentration bound, implicit in~\cite{bhmt:countingj}: 
$\Pr[|\tilde{t} - t| \geq jn/T] = O(1/j)$.  For example, using $T=O(\sqrt{n})$ quantum queries we can, with high probability, 
approximate $t$ up to additive error of $O(\sqrt{n})$.
\item {\bf Search on bounded-error inputs:} 
Suppose the bits $x_1,\ldots,x_n$ are not given by a perfect oracle $O_x$, but by an imperfect one:
$$
O_x:\ket{i,b,0}\mapsto\sqrt{1-\eps_i}\,\ket{i,b\oplus
  x_i,w_i}+\sqrt{\eps_i}\,\ket{i,\overline{b\oplus x_i},w'_i}\,,
$$ 
where we know $\eps$, we know that $\eps_i\leq\eps$ for each $x$ and $i$, but we do not know the actual values of the $\eps_i$
(which may depend on $x$), or of the ``workspace'' states $\ket{w_i}$ and $\ket{w'_i}$.
We call this an \emph{$\eps$-bounded-error quantum oracle}. 
This situation arises, for instance, when each bit $x_i$ is itself computed by some bounded-error quantum algorithm.
Given the ability to apply $O_x$ as well as its inverse $O_x^{-1}$,
we can still find a solution with high probability using $O(\sqrt{n})$ queries~\cite{hmw:berrorsearch}.
If the unknown number of solutions is $t$, then we can still find one with high probability using $O(\sqrt{n/t})$ queries.
\end{itemize}

\paragraph{From quantum query algorithms to polynomials}
An \emph{$n$-variate multilinear polynomial} $p$ is a function $p:\mathbb{C}^n\rightarrow\mathbb{C}$
that can be written as
$$
p(x_1,\ldots,x_n)=\sum_{S\subseteq[n]} a_S\prod_{i\in S}x_i\,,
$$
for some complex numbers $a_S$.
The \emph{degree} of $p$ is $\deg(p)=\max\{|S| : a_S\neq 0\}$.
It is well known (and easy to show) that every function $f:{\bitset}^n\rightarrow\mathbb{C}$
has a unique representation as such a polynomial; $\deg(f)$ is defined as the degree of that polynomial.
For example, the 2-bit AND function is $p(x_1,x_2)=x_1x_2$, and the 2-bit Parity function is $p(x_1,x_2)=x_1+x_2-2x_1x_2$.
Both polynomials have degree~2.

For the purposes of this survey, the crucial property of efficient quantum query algorithms is that the
amplitudes of their final state are low-degree polynomials of $x$~\cite{fortnow&rogers:limitations,bbcmw:polynomialsj}.
More precisely:

\begin{lemma}
Consider a $T$-query algorithm with input $x\in{\bitset}^n$ acting on an $m$-qubit space.
Then its final state can be written as
$$
\sum_{z\in{\bitset}^m}\alpha_z(x)\ket{z}\,,
$$
where each $\alpha_z$ is a multilinear polynomial in $x$ of degree at most $T$.
\end{lemma}

\begin{proof}
The proof is by induction on $T$. The base case ($T=0$) trivially holds: the algorithm's starting state is independent of $x$,
so its amplitudes are polynomials of degree~0.

For the induction step, note that a fixed linear transformation does not increase the degree of the amplitudes 
(the new amplitudes are linear combinations of the old amplitudes), while a query to $x$ corresponds to the following map:
$$
\alpha_{i,0,w}\ket{i,0,w}+\alpha_{i,1,w'}\ket{i,1,w'}\mapsto ((1-x_i)\alpha_{i,0,w}+x_i\alpha_{i,1,w'})\ket{i,0,w}+(x_i\alpha_{i,0,w}+(1-x_i)\alpha_{i,1,w'})\ket{i,1,w'}\,,
$$
which increases the degree of the amplitudes by at most~1:
if $\alpha_{i,0,w}$ and $\alpha_{i,1,w'}$ are polynomials in $x$ of degree at most $d$, 
then the new amplitudes are polynomials of degree at most $d+1$.
Since our inputs are 0/1-valued, we can drop higher degrees and 
assume without loss of generality that the resulting polynomials are multilinear.
\end{proof}

If we measure the first qubit of the final state and output the resulting bit, then the probability
of output~1 is given by 
$$
\sum_{\substack{z\in {\bitset}^m,\\z_1=1}}|\alpha_z|^2\,,
$$
which is a real-valued polynomial of $x$ of degree at most $2T$.  
This is true more generally:

\begin{corollary}\label{cor2Tpoly}
Consider a $T$-query algorithm with input $x\in{\bitset}^n$.
Then the probability of a specific output is a multilinear polynomial in $x$ of degree at most $2T$.
\end{corollary}

This connection between quantum query algorithms and polynomials has mostly been used as 
a tool for \emph{lower bounds}~\cite{bbcmw:polynomialsj,aaronson&shi:collision,aaronson:advicecommj,ksw:dpt-siam}:
if one can show that every polynomial that approximates a function $f:{\bitset}^n\rightarrow{\bitset}$ has degree at least $d$,
then every quantum algorithm computing $f$ with small error must use at least $d/2$ queries.  
We give one example in this spirit in Section~\ref{PPvsUPP}, in which a version of the polynomial method 
yielded a breakthrough in \emph{classical} lower bounds.
However, most of the applications in this survey (in Section~\ref{secapppoly}) work in the other direction: 
they view quantum algorithms as a means for constructing polynomials with certain desirable properties.

\section{Using quantum information theory}\label{secappqit}

The results in this section all use quantum information-theoretic bounds to say something about non-quantum objects.

\subsection{Communication lower bound for inner product}\label{ssecip}

The first surprising application of quantum information theory to another area was
in \emph{communication complexity}. 
The basic scenario in this area models 2-party distributed computation:
Alice receives some $n$-bit input $x$, Bob receives some $n$-bit input $y$,
and together they want to compute some Boolean function $f(x, y)$, the value of which Bob is required to output (with high probability, in the case of bounded-error protocols).
The resource to be minimized is the amount of communication between Alice
and Bob, whence the name communication complexity.  This model was introduced
by Yao~\cite{yao:distributive}, and a good overview of (non-quantum) results and applications 
may be found in the book of Kushilevitz and Nisan~\cite{kushilevitz&nisan:cc}.
The area is interesting in its own right as a basic complexity measure for distributed computing, 
but has also found many applications as a tool for lower bounds in areas like
data structures, Turing machine complexity, etc.
The quantum generalization is quite straightforward: now Alice and Bob can communicate qubits, and possibly start
with an entangled state. See~\cite{wolf:qccsurvey} for more details and a survey of results.

One of the most studied communication complexity problems is the \emph{inner product} problem,
where the function to be computed is the inner product of $x$ and $y$ modulo $2$, 
i.e., $\IP(x,y)=\sum_{i=1}^n x_iy_i \bmod 2$. Clearly, $n$ bits
of communication suffice for any function---Alice can just send $x$. However, $\IP$ is a good example where 
one can prove that nearly $n$ bits of communication is also \emph{necessary}.
The usual proof for this result is based on the combinatorial notion of ``discrepancy,''
but below we give an alternative quantum-based proof due to Cleve et al.~\cite{cdnt:ip}.

Intuitively, it seems that unless Alice gives Bob a lot of information about $x$,
he will not be able to guess the value of $\IP(x,y)$.
However, in general it is hard to directly lower bound communication complexity by information,
since we really require Bob to produce only one bit of output.%
\footnote{Still, there are also classical techniques to turn this information-theoretic
intuition into communication complexity  lower bounds~\cite{cswy:directsum,jks:2apps,bjks:itcc,bjks:datastream,jkr:readonce,leonardos&saks:readoncej}.}
The very elegant proof of~\cite{cdnt:ip} uses quantum effects to get around this problem:
it converts a protocol (quantum or classical) that computes $\IP$ into a quantum protocol
that communicates $x$ from Alice to Bob. Holevo's theorem then easily lower bounds 
the amount of communication of the latter protocol by the length of $x$.
This goes as follows.
Suppose Alice and Bob have some protocol for $\IP$, say it uses $c$ bits of communication.
Suppose for simplicity it has no error probability.
By running the protocol, putting the answer bit $x\cdot y$ into a phase, and then reversing the protocol 
to set its workspace back to its initial value, we can implement the following unitary mapping
$$
\ket{x}\ket{y}\mapsto\ket{x}(-1)^{x\cdot y}\ket{y}\,.
$$
Note that this protocol now uses $2c$ bits of communication: $c$ going from Alice to Bob and $c$ going from Bob to Alice. 
The trick is that we can run this unitary on a \emph{superposition} of inputs, at a cost of $2c$ \emph{qubits} of communication.
Suppose Alice starts with an arbitrary $n$-bit state $\ket{x}$
and Bob starts with the uniform superposition
$\frac{1}{\sqrt{2^n}}\sum_{y\in{\bitset}^n}\ket{y}$.
If they apply the above unitary, the final state becomes
$$
\ket{x}\frac{1}{\sqrt{2^n}}\sum_{y\in{\bitset}^n}(-1)^{x\cdot y}\ket{y}\,.
$$
If Bob now applies a Hadamard transform to each of his $n$ qubits,
then he obtains the basis state $\ket{x}$, so Alice's $n$ classical
bits have been communicated to Bob.
Theorem~\ref{thholevo} now implies that Alice must have sent at least $n/2$ qubits to Bob
(even if Alice and Bob started with unlimited shared entanglement). Hence $c\geq n/2$.

With some more technical complication, the same idea gives a linear lower bound
on the communication of bounded-error protocols for $\IP$.
Nayak and Salzman~\cite{nayak&salzman:entanglementj} later obtained 
optimal bounds for quantum protocols computing $\IP$.

\subsection{Lower bounds on locally decodable codes}

The development of error-correcting codes is one of the success stories of science in the second half of the 20th century.
Such codes are eminently practical, and are widely used to protect information stored on discs, communication over channels, etc.
From a theoretical perspective, there exist codes that are nearly optimal in a number of different respects simultaneously: 
they have constant rate, can protect against a constant noise-rate, and have linear-time encoding and decoding procedures.
We refer to Trevisan's survey~\cite{trevisan:eccsurvey} for a complexity-oriented discussion of codes and their applications.

One drawback of ordinary error-correcting codes is that we cannot efficiently decode small parts of the encoded
information. If we want to learn, say, the first bit of the encoded message then we usually still need to decode the whole encoded string.
This is relevant in situations where we have encoded a very large string (say, a library of books, or a large database),
but are only interested in recovering small pieces of it at any given time.
Dividing the data into small blocks and encoding each block separately will not work: 
small chunks will be efficiently decodable but not error-correcting, since a tiny fraction of 
well-placed noise could wipe out the encoding of one chunk completely. 
There exist, however, error-correcting codes that are \emph{locally decodable}, in the sense that we can efficiently recover
individual bits of the encoded string. These are defined as follows~\cite{katz&trevisan:ldc}:

\begin{definition}
$C:{\bitset}^n\rightarrow{\bitset}^m$ is a \emph{$(q,\delta,\eps)$-locally decodable code}
(LDC) if there is a classical randomized decoding algorithm $A$ such that
\begin{enumerate}
\item $A$ makes at most $q$ queries to an $m$-bit string $y$ (non-adaptively).
\item For all $x$ and $i$, and all $y\in{\bitset}^m$ with Hamming distance
$d(C(x),y)\leq\delta m$ we have $$
\Pr[A^y(i)=x_i]\geq 1/2+\eps\,.
$$
\end{enumerate}
\end{definition}

The notation $A^y(i)$ reflects that the decoder $A$ has two different types of input.
On the one hand there is the (possibly corrupted) codeword $y$, to which the decoder has oracle access and from which it can read at most $q$
bits of its choice.  On the other hand there is the index $i$ of the bit that needs to be recovered, which is known fully to the decoder.

The main question about LDCs is the tradeoff between the codelength $m$ and the number of queries $q$ (which is a proxy for the decoding-time).
This tradeoff is still not very well understood.  
We list the best known constructions here.
On one extreme, regular error-correcting codes are $(m,\delta,1/2)$-LDCs, so one can have LDCs of linear length if one allows a linear number of queries.
Reed-Muller codes allow one to construct LDCs with $m=\poly(n)$ and $q=\polylog(n)$~\cite{cgks:pir}.
For constant $q$, the best constructions are due to Efremenko~\cite{efremenko:ldc}, improving upon Yekhanin~\cite{yekhanin:3ldcj}:
for $q=2^r$ one can get codelength roughly $2^{2^{(\log n)^{1/r}}}$, and for $q=3$ one gets roughly $2^{2^{\sqrt{\log n}}}$.
For $q=2$ there is the Hadamard code: given $x\in{\bitset}^n$, define a codeword of length $m=2^n$ by writing down
the bits $x\cdot z \bmod 2$, for all $z\in{\bitset}^n$. One can decode $x_i$ with 2 queries as follows: choose $z\in{\bitset}^n$ uniformly at random
and query the (possibly corrupted) codeword at indices $z$ and $z\oplus e_i$, where the latter denotes the string obtained from $z$ by flipping its $i$-th bit. 
Individually, each of these two indices is uniformly
distributed.  Hence for each of them, the probability that the returned bit of is corrupted is at most $\delta$. By the union bound, with probability
at least $1-2\delta$, both queries return the uncorrupted values. Adding these two bits modulo $2$ gives the correct answer:
$$
C(x)_z\oplus C(x)_{z\oplus e_i}=(x\cdot z)\oplus (x\cdot (z\oplus e_i))=x\cdot e_i=x_i\,.
$$
Thus the Hadamard code is a $(2,\delta,1/2-2\delta)$-LDC of exponential length.
Can we still do something if we can make only one query instead of two?
It turns out that 1-query LDCs do not exist once $n$ is sufficiently large~\cite{katz&trevisan:ldc}.

The only superpolynomial \emph{lower bound} known on the length of LDCs is for the case of 2 queries: 
there one needs an exponential codelength and hence the Hadamard code is essentially optimal. This was first shown
for \emph{linear} 2-query LDCs by Goldreich et al.~\cite{gkst:lowerpir} via a combinatorial argument, and
then for general LDCs by Kerenidis and de Wolf~\cite{kerenidis&wolf:qldcj} via a \emph{quantum} argument.%
\footnote{The best known lower bounds for general LDCs with $q>2$ queries are only slightly superlinear.  
Those bounds, and also the best known lower bounds for 2-server \emph{Private Information Retrieval} schemes, 
are based on similar quantum ideas~\cite{kerenidis&wolf:qldcj,wehner&wolf:improvedldc,woodruff:ldclower}.
The best known lower bound for 3-query LDCs is $m=\Omega(n^2/\log n)$~\cite{woodruff:ldclower};
for \emph{linear} 3-query LDCs, a slightly better lower bound of $m=\Omega(n^2)$ 
is known~\cite{woodruff:3querylower}.}
The easiest way to present this argument is to assume the following fact, which states a kind of
``normal form'' for the decoder.

\begin{fact}[Katz and Trevisan~\cite{katz&trevisan:ldc} + folklore]\label{ldcfact}
For every $(q,\delta,\eps)$-LDC $C:{\bitset}^n\rightarrow{\bitset}^m$, and for each $i\in[n]$, there exists a set $M_i$ of $\Omega(\delta\eps m/q^2)$ disjoint tuples,
each of at most $q$ indices from $[m]$, and a bit $a_{i,t}$ for each tuple $t\in M_i$, such that the following holds:
\begin{equation} \label{eqldcfact}
\Pr_{x\in{\bitset}^n}\left[x_i = a_{i,t}\oplus\sum_{j\in t}C(x)_j\right]\geq 1/2+\Omega(\eps/2^q)\,,
\end{equation}
where the probability is taken uniformly over $x$.
Hence to decode $x_i$ from $C(x)$, the decoder can just query the indices in a randomly chosen tuple $t$ from $M_i$,
outputting the sum of those $q$ bits and $a_{i,t}$.
\end{fact}

\noindent
Note that the above decoder for the Hadamard code is already of this form, with $M_i=\{(z,z\oplus e_i)\}$.  
We omit the proof of Fact~\ref{ldcfact}.  It uses purely classical ideas and is not hard. 

Now suppose $C:{\bitset}^n\rightarrow{\bitset}^m$ is a $(2,\delta,\eps)$-LDC.  We want to show that the codelength $m$ must be exponentially large in $n$.
Our strategy is to show that the following $m$-dimensional quantum encoding is in fact a quantum random access code for $x$,
with some success probability $p>1/2$:
$$
x\mapsto \ket{\phi_x}=\frac{1}{\sqrt{m}} \sum_{j=1}^m (-1)^{C(x)_j}\ket{j}\,.
$$
Theorem~\ref{thrac} then implies that the number of qubits of this state (which is $\ceil{\log m}$) is at least $(1-H(p))n=\Omega(n)$, and we are done.

Suppose we want to recover $x_i$ from $\ket{\phi_x}$.
We turn each $M_i$ from Fact~\ref{ldcfact} into a measurement:
for each pair $(j,k)\in M_i$ form the projector $P_{jk}=\ketbra{j}{j}+\ketbra{k}{k}$,
and let $P_{rest}=\sum_{j\not\in\cup_{t\in M_i}t}\ketbra{j}{j}$ be the projector on the remaining indices.
These $|M_i|+1$ projectors sum to the $m$-dimensional identity matrix, so they form a valid projective measurement.
Applying this to $\ket{\phi_x}$ gives outcome $(j,k)$ with probability $2/m$ for each $(j,k)\in M_i$, and outcome ``rest'' with probability $\rho=1-\Omega(\delta\eps)$.
In the latter case we just output a fair coin flip as our guess for $x_i$.  In the former case the state has collapsed to the
following useful superposition:
$$
\frac{1}{\sqrt{2}}\left((-1)^{C(x)_j}\ket{j}+(-1)^{C(x)_k}\ket{k}\right)=\frac{(-1)^{C(x)_j}}{\sqrt{2}}\left(\ket{j}+(-1)^{C(x)_j\oplus C(x)_k}\ket{k}\right)
$$
Doing a 2-outcome measurement in the basis $({1}/{\sqrt{2}})\,(\ket{j}\pm\ket{k})$ now gives us the value $C(x)_j\oplus C(x)_k$ with probability~1.
By~\eqref{eqldcfact}, if we add the bit $a_{i,(j,k)}$ to this, we get $x_i$ with probability at least $1/2+\Omega(\eps)$.
The success probability of recovering $x_i$, averaged over all $x$, is 
$$
p\geq \frac{1}{2}\rho + \left(\frac{1}{2}+\Omega(\eps)\right)(1-\rho)=\frac{1}{2}+\Omega(\delta\eps^2)\,.
$$
Now $1-H(1/2+\eta)=\Theta(\eta^2)$ for $\eta\in[0,1/2]$, so after applying Theorem~\ref{thrac} we obtain the following:

\begin{theorem}[Kerenidis and de Wolf]
If $C:{\bitset}^n\rightarrow{\bitset}^m$ is a $(2,\delta,\eps)$-locally decodable code, then $m=2^{\Omega(\delta^2\eps^4 n)}$.
\end{theorem}

The dependence on $\delta$ and $\eps$ in the exponent can be improved to $\delta\eps^2$~\cite{kerenidis&wolf:qldcj}.
This is still the only superpolynomial lower bound known for LDCs.
An alternative proof was found later~\cite{brw:hypercontractive}, using an extension of the Bonami-Beckner hypercontractive inequality.
However, even that proof still follows the outline of the above quantum-inspired proof, albeit in linear-algebraic language.

\subsection{Rigidity of Hadamard matrices}

In this section we describe an application of quantum information theory to matrix theory from~\cite{wolf:rigidity}.
Suppose we have some $n\times n$ matrix $M$, whose rank we want to reduce by changing a few entries.  The \emph{rigidity} of $M$ measures the minimal number of entries we need to change in order to reduce its rank to a given value $r$. This notion can be studied over any field, but we will focus 
here on $\mathbb{R}$ and $\mathbb{C}$.  Formally:
\begin{definition}
The \emph{rigidity} of a matrix $M$ is the following function:
$$
R_M(r)=\min\{d(M,\widetilde{M}) : \rank(\widetilde{M})\leq r\}\,,
$$
where $d(M,\widetilde{M})$ counts the Hamming distance, i.e., the number of coordinates where $M$ and $\widetilde{M}$ differ.
The \emph{bounded rigidity} of $M$ is defined as
$$
R_M(r,\theta)=\min\{d(M,\widetilde{M}) : \rank(\widetilde{M})\leq r, \max_{x,y}|M_{x,y}-\widetilde{M}_{x,y}|\leq\theta\}\,.
$$
\end{definition}

Roughly speaking, high rigidity means that $M$'s rank is robust:
changes in few entries will not change the rank much.
Rigidity was defined by Valiant~\cite[Section~6]{valiant:rigidity} in the 1970s
with a view to proving circuit lower bounds.  In particular,
he showed that an explicit $n\times n$ matrix $M$ with $R_M(\eps n)\geq n^{1+\delta}$
for $\eps,\delta>0$ would imply that log-depth arithmetic circuits (with linear gates) 
that compute the linear map $M:\mathbb{R}^n\rightarrow\mathbb{R}^n$ need
superlinear circuit size. This motivates trying to prove lower bounds on rigidity for specific matrices.
Clearly, $R_M(r)\geq n-r$ for every full-rank
matrix $M$, since reducing the rank by~1 requires changing at least one entry.
This bound is optimal for the identity matrix, but usually far from tight.
Valiant showed that most matrices have rigidity $(n-r)^2$, but
finding an \emph{explicit} matrix with high rigidity has been open for decades.%
\footnote{Lokam~\cite{lokam:quadratic} recently found an explicit $n\times n$ matrix with near-maximal rigidity; 
unfortunately his matrix has fairly large, irrational entries, and is not sufficiently explicit for Valiant's purposes.}
Similarly, finding explicit matrices with strong lower bounds on \emph{bounded} rigidity would 
have applications to areas like communication complexity and learning theory~\cite{lokam:rigidity,lokam:linalgsurvey}. 

A very natural and widely studied class of candidates for such a high-rigidity matrix
are the \emph{Hadamard matrices}. A Hadamard matrix is an $n\times n$ matrix $M$ with entries $+1$ and $-1$
that is orthogonal (so $M^T M=nI$).
Ignoring normalization, the $k$-fold tensor product of the matrix from~\eqref{eqhadamard} is a Hadamard matrix with $n=2^k$.
(It is a longstanding conjecture that Hadamard matrices exist if, and only if, $n$ equals~2 or a multiple of~4.)

Suppose we have a matrix $\widetilde{M}$ differing from the Hadamard matrix $M$ in
$R$ positions such that $\rank(\widetilde{M})\leq r$. The goal in proving
high rigidity is to lower-bound $R$ in terms of $n$ and $r$.
Alon~\cite{alon:rigidity} proved $R=\Omega(n^2/r^2)$.  This was later reproved
by Lokam~\cite{lokam:rigidity} using spectral methods.
Kashin and Razborov~\cite{kashin&razborov:rigidity} improved this to
$R\geq n^2/256r$.  De Wolf~\cite{wolf:rigidity} later rederived this bound using a quantum argument, with a better constant.  We present this argument next.

\paragraph{The quantum idea}
The idea is to view the rows of an $n\times n$ matrix as a quantum encoding of $[n]$.
The rows of a Hadamard matrix $M$, after normalization by a factor $1/\sqrt{n}$, 
form an orthonormal set of $n$-dimensional quantum states $\ket{M_i}$.
If Alice sends Bob $\ket{M_i}$ and Bob measures the received state with 
the projectors $P_j=\ketbra{M_j}{M_j}$, then he learns $i$ with probability~1,
since $|\braket{M_i}{M_j}|^2=\delta_{i,j}$.
Of course, nothing spectacular has been achieved by this---we just transmitted $\log n$ bits of information by sending $\log n$ qubits.

Now suppose that instead of $M$ we have some rank-$r$ $n\times n$ matrix
$\widetilde{M}$ that is ``close'' to $M$ (we are deliberately being vague about ``close'' here, 
since two different instantiations of the same idea apply to the two versions of rigidity).
Then we can still use the quantum states $\ket{\widetilde{M}_i}$
that correspond to its normalized rows.
Alice now sends the normalized $i$-th row of $\widetilde{M}$ to Bob.  Crucially,
she can do this by means of an $r$-dimensional quantum state, as follows.
Let $\ket{v_1},\ldots,\ket{v_r}$ be an orthonormal basis
for the row space of $\widetilde{M}$. In order to transmit the normalized $i$-th row
$\ket{\widetilde{M}_i}=\sum_{j=1}^r \alpha_j\ket{v_j}$, Alice sends
$\sum_{j=1}^r \alpha_j\ket{j}$ and Bob applies a unitary that maps
$\ket{j}\mapsto\ket{v_j}$ to obtain $\ket{\widetilde{M}_i}$.
He measures this with the projectors $\{P_j\}$.  Then his probability
of getting the correct outcome $i$ is
$$
p_i=|\inp{M_i}{\widetilde{M}_i}|^2\,.
$$
The ``closer'' $\widetilde{M}$ is to $M$, the higher these $p_i$'s are.
But~\eqref{nayakpibound} in Section~\ref{ssecqinfo} 
tells us that the sum of the $p_i$'s lower-bounds the dimension $r$ of the quantum system.
Accordingly, the ``closer'' $\widetilde{M}$ is to $M$, the higher its rank has to be.
This is exactly the tradeoff that rigidity tries to measure.

This quantum approach allows us to quite easily derive Kashin and Razborov's~\cite{kashin&razborov:rigidity} bound on rigidity, with a better constant.

\begin{theorem}[de Wolf, improving Kashin and Razborov]\label{thrigid}
Let $M$ be an $n\times n$ Hadamard matrix.
If $r\leq n/2$, then $R_M(r)\geq n^2/4r$.
\end{theorem}

Note that if $r\geq n/2$ then $R_M(r)\leq n$, at least for symmetric Hadamard matrices such as $H^{\otimes k}$:
then $M$'s eigenvalues are all $\pm\sqrt{n}$, so we can reduce its rank to $n/2$ or less by
adding or subtracting the diagonal matrix $\sqrt{n}I$. Hence a superlinear lower bound on $R_M(r)$ cannot be proved for $r\geq n/2$.

\medskip

\begin{proof}
Consider a rank-$r$ matrix $\widetilde{M}$ differing from $M$ in $R=R_M(r)$ entries.
By averaging, there exists a set of $a=2r$ rows of $\widetilde{M}$ with a total number of at most
$aR/n$ errors (i.e., changes compared to $M$). Now consider the submatrix $A$ of $\widetilde{M}$ consisting
of those $a$ rows and the $b\geq n-aR/n$ columns that have no errors in those $a$ rows.
If $b=0$ then $R\geq n^2/2r$ and we are done, so we can assume $A$ is nonempty.
This $A$ is error-free, hence a submatrix of $M$ itself.
We now use the quantum idea to prove the following claim 
(originally proved by Lokam using linear algebra, see the end of this section):

\begin{claim}[Lokam]\label{claimlokam}
Every $a\times b$ submatrix $A$ of $n\times n$ Hadamard matrix $M$ has rank $r\geq ab/n$.
\end{claim}

\begin{proof}
Obtain the rank-$r$ matrix $\widetilde{M}$ from $M$ by setting all entries outside of $A$ to 0.
Consider the $a$ quantum states $\ket{\widetilde{M}_i}$ corresponding to the nonempty rows;
they have normalization factor $1/\sqrt{b}$.  Alice tries to communicate a value $i \in [a]$ to Bob by sending $\ket{\widetilde{M}_i}$.  For each such $i$,
Bob's probability of successfully decoding $i$ is $p_i=|\inp{M_i}{\widetilde{M}_i}|^2=|b/\sqrt{bn}|^2=b/n.$
The states $\ket{\widetilde{M}_i}$ are all contained in an $r$-dimensional space, so \eqref{nayakpibound} implies $\sum_{i=1}^a p_i\leq r$.  Combining both bounds concludes the proof.
\end{proof}

Hence we get
$$
r=\rank(\widetilde{M})\geq \rank(A)\geq \frac{ab}{n}\geq
\frac{a(n-aR/n)}{n}\,.
$$
Rearranging gives the theorem.
\end{proof}

Applying the quantum idea in a different way allows us to also analyze \emph{bounded} rigidity:

\begin{theorem}[Lokam, Kashin and Razborov, de Wolf]\label{thboundedrigid}
Let $M$ be an $n\times n$ Hadamard matrix and $\theta>0$.
Then 
$$
R_M(r,\theta)\geq\frac{n^2(n-r)}{2\theta n+r(\theta^2+2\theta)}\,.
$$
\end{theorem}

\begin{proof}
Consider a rank-$r$ matrix $\widetilde{M}$ differing from $M$ in $R=R_M(r,\theta)$ entries, with each entry $\widetilde{M}_{ij}$ differing from $M_{ij}$ by at most $\theta$.
As before, define quantum states corresponding to its rows:
$\ket{\widetilde{M}_i}=c_i\sum_{j=1}^n \widetilde{M}_{i,j}\ket{j}$,
where 
$$
c_i=1/\sqrt{\sum_j |\widetilde{M}_{i,j}|^2}
$$
is a normalizing constant.
Note that 
$$
\sum_j |\widetilde{M}_{i,j}|^2\leq (n-d(M_i,\widetilde{M}_i))+d(M_i,\widetilde{M}_i)(1+\theta)^2=n+d(M_i,\widetilde{M}_i)(\theta^2+2\theta)\,,
$$
where $d(\cdot,\cdot)$ measures Hamming distance.  Alice again sends $\ket{\widetilde{M}_i}$ to Bob to communicate the value $i \in [a]$.  Bob's success probability $p_i$ is now
$$
p_i = |\inp{M_i}{\widetilde{M}_i}|^2
    \geq \frac{c_i^2}{n}(n-\theta d(M_i,\widetilde{M}_i))^2
    \geq c_i^2(n-2\theta d(M_i,\widetilde{M_i}))
    \geq \frac{n-2\theta d(M_i,\widetilde{M}_i)}{n+d(M_i,\widetilde{M}_i)(\theta^2+2\theta)}\,.
$$
Observe that our lower bound on $p_i$ is a convex function of the Hamming distance $d(M_i,\widetilde{M}_i)$.  Also, $\Exp[d(M_i,\widetilde{M}_i)] = R/n$ over a uniform choice of $i$.  Therefore by Jensen's inequality we obtain the lower bound for the average success probability $p$ when $i$ is uniform:
$$
p\geq \frac{n-2\theta R/n}{n+R(\theta^2+2\theta)/n}\,.
$$
Now \eqref{nayakpibound} implies $p\leq r/n$.  Combining and rearranging gives the theorem.
\end{proof}

\noindent
For $\theta\geq n/r$ we obtain the second result of Kashin and Razborov~\cite{kashin&razborov:rigidity}:
$$
R_M(r,\theta)=\Omega(n^2(n-r)/r\theta^2)\,.
$$
If $\theta\leq n/r$ we get an earlier result of Lokam~\cite{lokam:rigidity}:
$$
R_M(r,\theta)=\Omega(n(n-r)/\theta)\,.
$$

\paragraph{Did we need quantum tools for this?}
Apart from Claim~\ref{claimlokam} the proof of Theorem~\ref{thrigid} is fully classical, 
and that claim itself can quite easily be proved using linear algebra, as was done originally by Lokam~\cite[Corollary~2.2]{lokam:rigidity}.
Let $\sigma_1(A),\ldots,\sigma_r(A)$ be the singular values of rank-$r$ submatrix $A$.
Since $M$ is an orthogonal matrix we have $M^T M=n I$, so all $M$'s singular values equal $\sqrt{n}$.
The matrix $A$ is a submatrix of $M$, so all $\sigma_i(A)$ are at most $\sqrt{n}$.
Using the Frobenius norm, we obtain the claim:
$$
ab=\norm{A}_F^2=\sum_{i=1}^r\sigma_i(A)^2\leq rn\,.
$$
Furthermore, after reading a first version of~\cite{wolf:rigidity}, Midrijanis~\cite{midrijanis:rigidity} 
came up with an even simpler proof of the $n^2/4r$ bound on rigidity for the special case of 
$2^k\times 2^k$ Hadamard matrices that are the $k$-fold tensor product of the $2\times 2$ Hadamard matrix.

In view of these simple non-quantum proofs, one might argue that the quantum approach is an overkill here.
However, the main point here was not to rederive more or less known bounds,
but to show how quantum tools provide a quite different perspective on the problem:
we can view a rank-$r$ approximation of the Hadamard matrix as a way
of encoding $[n]$ in an $r$-dimensional quantum system; quantum information-theoretic 
bounds such as \eqref{nayakpibound} can then be invoked
to obtain a tradeoff between the rank $r$ and the ``quality'' of the approximation.
The same idea was used to prove Theorem~\ref{thboundedrigid}, whose proof cannot be so easily de-quantized.
The hope is that this perspective may help in the future to settle some of the longstanding open problems about rigidity.

\section{Using the connection with polynomials}\label{secapppoly}

The results of this section are based on the connection explained at the end of Section~\ref{ssecqueryalgos}:
efficient quantum query algorithms give rise to low-degree polynomials.

As a warm-up, we mention a recent application of this.
A \emph{formula} is a binary tree whose internal nodes are AND and OR-gates, and each leaf is a Boolean input variable $x_i$ or its negation.
The root of the tree computes a Boolean function of the input bits in the obvious way.
The size of the formula is its number of leaves.
O'Donnell and Servedio~\cite{odonnell&servedio:newdegree} conjectured that all formulas of size $n$ have \emph{sign-degree} at most $O(\sqrt{n})$;
the sign-degree is the minimal degree among all $n$-variate polynomials that are positive if, and only if, the formula is~1.
Their conjecture implies, by known results, that the class of formulas is \emph{learnable} in the PAC model in time $2^{n^{1/2+o(1)}}$.

Building on a quantum algorithm of Farhi et al.~\cite{fgg:nandtreej} that was inspired by physical notions from scattering theory, 
Ambainis et al.~\cite{acrsz:andorj} showed that for every formula there is a quantum algorithm that computes it using $n^{1/2+o(1)}$ queries.
By Corollary~\ref{cor2Tpoly}, the acceptance probability of this algorithm 
is an approximating polynomial for the formula, of degree $n^{1/2+o(1)}$.
Hence that polynomial minus 1/2 is a sign-representing polynomial for the formula, 
proving the conjecture of O'Donnell and Servedio up to the $o(1)$ in the exponent.
Based on an improved $O(\sqrt{n}\log(n)/\log\log(n))$-query quantum
algorithm by Reichardt~\cite{reichardt:tight} and some additional analysis, Lee~\cite{lee:signdegree} subsequently
improved this general upper bound on the sign-degree of formulas to the optimal $O(\sqrt{n})$,
fully proving the conjecture (in contrast to~\cite{acrsz:andorj}, he really bounds sign-degree, not approximate degree).

\subsection{$\eps$-approximating polynomials for symmetric functions}

Our next example comes from~\cite{wolf:degreesymmf}, and deals with the minimal degree of
$\eps$-approximating polynomials for \emph{symmetric} Boolean functions.
A function $f:{\bitset}^n\rightarrow{\bitset}$ is symmetric if its value only depends on the Hamming weight
$|x|$ of its input $x\in{\bitset}^n$. Equivalently, 
$f(x)=f(\pi(x))$ for all $x\in{\bitset}^n$ and all permutations $\pi\in S_n$.
Examples are OR, AND, Parity, and Majority.

For some specified approximation error $\eps$, let $\deg_{\eps}(f)$ denote 
the minimal degree among all $n$-variate multilinear polynomials $p$ satisfying $|p(x)-f(x)|\leq\eps$ for all $x\in{\bitset}^n$.
If one is interested in constant error then one typically fixes $\eps=1/3$, 
since approximations with different constant errors can easily be converted into each other.
Paturi~\cite{paturi:degree} tightly characterized the 1/3-error approximate degree: 
if $t\in(0,n/2]$ is the smallest integer such that $f$ is constant for $|x|\in\{t,\ldots,n-t\}$, 
then $\deg_{1/3}(f)=\Theta(\sqrt{tn})$.

Motivated by an application to the inclusion-exclusion principle of probability theory,
Sherstov~\cite{sherstov:inclexcl} recently studied the dependence of the degree on the error $\eps$.
He proved the surprisingly clean result that for all $\eps\in[2^{-n},1/3]$,
$$
\deg_{\eps}(f)=\widetilde{\Theta}\left(\deg_{1/3}(f) + \sqrt{n\log(1/\eps)}\right),
$$
where the $\widetilde{\Theta}$ notation hides some logarithmic factors 
(note that the statement is false if $\eps\ll 2^{-n}$, since clearly $\deg(f)\leq n$ for all $f$.)
His upper bound on the degree is based on Chebyshev polynomials.  De Wolf~\cite{wolf:degreesymmf} tightens this upper bound on the degree:

\begin{theorem}[de Wolf, improving Sherstov]\label{thepsdeg}
For every non-constant symmetric function $f:{\bitset}^n\rightarrow{\bitset}$ and $\eps\in[2^{-n},1/3]$:
$$
\deg_{\eps}(f)=O\left(\deg_{1/3}(f) + \sqrt{n\log(1/\eps)}\right).
$$
\end{theorem}

By the discussion at the end of Section~\ref{ssecqueryalgos}, to prove Theorem~\ref{thepsdeg}
it suffices to give an $\eps$-error quantum algorithm for $f$ that uses $O(\deg_{1/3}(f) + \sqrt{n\log(1/\eps)})$ queries. 
The probability that the algorithm outputs~1 will be our $\eps$-error polynomial.
For example, the special case where $f$ is the $n$-bit OR function follows immediately from the $O(\sqrt{n\log(1/\eps)})$-query
search algorithm with error probability $\eps$ that was mentioned there.

Here is the algorithm for general symmetric $f$.  It uses some of the algorithms listed in Section~\ref{ssecqueryalgos} as subroutines.  Let $t = t(f)$ be as in Paturi's result.
\begin{enumerate}
\item Use $t - 1$ applications of exact Grover to try to find up to $t - 1$ distinct solutions in $x$
(remember that a ``solution'' to the search problem is an index $i$ such that $x_i=1$).
Initially we run an exact Grover assuming $|x|=t - 1$, we verify that the outcome is a solution at the expense of one more query, 
and then we ``cross it out'' to prevent finding the same solution again in subsequent searches. 
Then we run another exact Grover assuming there are $t-2$ solutions, etc.  
Overall, this costs 
$$
\sum_{i=1}^{t - 1} O(\sqrt{n/i})=O(\sqrt{tn})=O(\deg_{1/3}(f))
$$
queries. 
\item Use $\eps/2$-error Grover to try to find one more solution. 
This costs $O(\sqrt{n\log(1/\eps)})$ queries.
\item The same as step~1, but now looking for positions of 0s instead of 1s.
\item The same as step~2, but now looking for positions of 0s instead of 1s.
\item If step~2 did not find another 1, then we assume step~1 found all 1s (i.e., a complete description of $x$), and we output the corresponding value of $f$.

Else, if step~4 did not find another 0, then we assume step~3 found all 0s, and we output the corresponding value of $f$.

Otherwise, we assume $|x|\in\{t,\ldots,n-t\}$ and output the corresponding value of $f$.
\end{enumerate}
Clearly the query complexity of this algorithm is $O(\deg_{1/3}(f) + \sqrt{n\log(1/\eps)})$, so it remains to upper bound its error probability.
If $|x|<t$ then step~1 finds all 1s with certainty and step~2 will not find another 1 (since there aren't any left after step~1), so in this case
the error probability is~0.  If $|x|>n-t$ then step~2 finds a 1 with probability at least $1-\eps/2$, step~3 finds all 0s with certainty,
and step~4 does not find another 0 (again, because there are none left); hence in this case the error probability is at most $\eps/2$.  
Finally, if $|x|\in\{t,\ldots,n-t\}$ then with probability at least $1-\eps/2$ step~2 will find another 1, 
and with probability at least $1-\eps/2$ step~4 will find another 0.
Thus with probability at least $1-\eps$ we correctly conclude $|x|\in\{t,\ldots,n-t\}$ and output the correct value of $f$.
Note that the only property of $f$ used here is that $f$ is constant on $|x|\in\{t,\ldots,n-t\}$;
the algorithm still works for Boolean functions $f$ that are arbitrary (non-symmetric) when $|x|\not\in\{t,\ldots,n-t\}$,
with the same query complexity $O(\sqrt{tn} + \sqrt{n\log(1/\eps)})$.

\subsection{Robust polynomials}
In the previous section we saw how quantum query algorithms allow us to construct polynomials 
(of essentially minimal degree) that $\eps$-approximate symmetric Boolean functions.
In this section we show how to construct \emph{robust} polynomial approximations.
These are insensitive to small changes in their $n$ input variables.
Let us first define more precisely what we mean:

\begin{definition}
Let $p:\mathbb{R}^n\rightarrow\mathbb{R}$ be an $n$-variate polynomial (not necessarily multilinear).
Then $p$ \emph{$\eps$-robustly} approximates $f:{\bitset}^n\rightarrow{\bitset}$ if for every $x\in{\bitset}^n$ and every $z\in[0,1]^n$ 
satisfying $|z_i-x_i|\leq\eps$ for all $i\in[n]$, we have $p(z)\in[0,1]$ and $|p(z)-f(x)|\leq\eps$.
\end{definition}

Note that we do not restrict $p$ to be multilinear, since the inputs we care about are no longer 0/1-valued.
The \emph{degree} of $p$ is its total degree. Note that we require both the $z_i$'s and the value $p(z)$ to be in the interval $[0,1]$.
This is just a matter of convenience, because it allows us to interpret these numbers as probabilities; 
using the interval $[-\eps,1+\eps]$ instead of $[0,1]$ would give an essentially equivalent definition.

One advantage of the class of robust polynomials over the usual approximating polynomials,
is that it is closed under composition: plugging robust polynomials into a robust polynomial gives another robust polynomial.
For example, suppose a function $f:{\bitset}^{n_1n_2}\rightarrow{\bitset}$ is obtained by composing 
$f_1:{\bitset}^{n_1}\rightarrow{\bitset}$ with $n_1$ independent copies of $f_2:{\bitset}^{n_2}\rightarrow{\bitset}$
(for instance an AND-OR tree).  Then we can just compose an $\eps$-robust polynomial for $f_1$ of degree $d_1$
with an $\eps$-robust polynomial for $f_2$ of degree $d_2$, to obtain an $\eps$-robust polynomial for $f$ of degree $d_1d_2$.
The errors ``take care of themselves,'' in contrast to ordinary
approximating polynomials, which may not compose in this fashion.%
\footnote{Reichardt~\cite{reichardt:composition} showed recently that such a clean composition result also holds for 
the usual bounded-error quantum query complexity, by going back and forth between quantum algorithms and span programs (which compose cleanly).} 

How hard is it to construct robust polynomials?  
In particular, does their degree have to be much larger than the usual approximate degree?
A good example is the $n$-bit Parity function. 
If the $n$ inputs $x_1,\ldots,x_n$ are $0/1$-valued then the following polynomial represents Parity:%
\footnote{If inputs and outputs were $\pm 1$-valued, the polynomial would just be the product of the $n$ variables.}
\begin{equation}\label{eqparitypoly}
p(x)=\frac{1}{2}-\frac{1}{2}\prod_{i=1}^n (1-2x_i)\,.
\end{equation}
This polynomial has degree $n$, and it is known that any $\eps$-approximating polynomial for Parity needs degree $n$ as well.
However, it is clear that this polynomial is not robust: if each $x_i=0$ is replaced by $z_i=\eps$, then the resulting value $p(z)$ is exponentially close to 1/2 rather than $\eps$-close to the correct value 0.
One way to make it robust is to individually ``amplify'' each input variable $z_i$, such that 
if $z_i\in[0,\eps]$ then its amplified version is in, say, $[0,1/100n]$
and if $z_i\in[1-\eps,1]$ then its amplified version is in $[1-1/100n,1]$.
The following univariate polynomial of degree $k$ does the trick:
$$
a(y)=\sum_{j>k/2}\binom{k}{j}y^j(1-y)^{k-j}\,.
$$
Note that this polynomial describes the probability that $k$ coin flips, each with probability $y$ of being~1, have majority~1.
By standard Chernoff bounds, if $y\in[0,\eps]$ then $a(y)\in [0,\exp(-\Omega(k))]$ and if $y\in[1-\eps,1]$ then $a(y)\in [1-\exp(-\Omega(k)),1]$.
Taking $k=O(\log n)$ and substituting $a(z_i)$ for $x_i$ in \eqref{eqparitypoly} gives an $\eps$-robust polynomial for Parity of degree $\Theta(n\log n)$.
Is this optimal?  Since Parity crucially depends on each of its $n$ variables, and amplifying each $z_i$ to polynomially small error 
needs degree $\Theta(\log n)$, one might conjecture robust polynomials for Parity need degree $\Omega(n\log n)$.
Surprisingly, this is not the case: there exist $\eps$-robust polynomials for Parity of degree $O(n)$.
Even more surprisingly, the only way we know how to construct such robust polynomials is via the connection with quantum algorithms.
Based on the quantum search algorithm for bounded-error inputs mentioned in Section~\ref{ssecqueryalgos}, 
Buhrman et al.~\cite{bnrw:robustqj} showed the following:

\begin{theorem}[BNRW]\label{throbustalgo}
There exists a quantum algorithm that makes $O(n)$ queries to an $\eps$-bounded-error quantum oracle 
and outputs $x_1,\ldots,x_n$ with probability at least $1-\eps$.
\end{theorem}

The constant in the $O(\cdot)$ depends on $\eps$, but we will not write this dependence explicitly.

\medskip

\begin{proof} [Proof (sketch)]
The idea is to maintain an $n$-bit string $\widetilde{x}$, initially all-0,
and to look for differences between $\widetilde{x}$ and $x$. Initially this number of differences is $|x|$.
If there are $t$ difference points (i.e., $i\in[n]$ where $x_i\neq\widetilde{x}_i$), 
then the quantum search algorithm $A$ with bounded-error inputs 
finds a difference point $i$ with high probability using $O(\sqrt{n/t})$ queries. 
We flip the $i$-th bit of $\widetilde{x}$. If the search indeed yielded a difference point,
then this reduces the distance between $\widetilde{x}$ and~$x$ by one.
Once there are no differences left, we have $\widetilde{x}=x$, which we can verify by one more run of $A$.
If $A$ only finds difference points, then we would find all differences in total number of queries
$$
\sum_{t=1}^{|x|}O(\sqrt{n/t})=O(\sqrt{|x|n})\,.
$$
The technical difficulty is that $A$ errs (i.e., produces an output $i$ where actually $x_i=\widetilde{x}_i$)
with constant probability, and hence we sometimes increase rather than decrease 
the distance between $\widetilde{x}$ and~$x$.  The proof details in~\cite{bnrw:robustqj} show that the procedure 
is still expected to make progress, and with high probability finds all differences after $O(n)$ queries.\footnote{The same idea would 
work with \emph{classical} algorithms, but gives query complexity roughly $\sum_{t=1}^{|x|} n/t\sim n \ln |x|$.}
\end{proof}

This algorithm implies that we can compute, with $O(n)$ queries and error probability $\eps$, any Boolean function $f:{\bitset}^n\rightarrow{\bitset}$
on $\eps$-bounded-error inputs: just compute $x$ and output $f(x)$.
This is not true for \emph{classical} algorithms running on bounded-error inputs.
In particular, classical algorithms that compute Parity with such a noisy oracle need $\Theta(n\log n)$ queries~\cite{feige&al:computingWithNoisyInformation}.

The above algorithm for $f$ is ``robust'' in a very similar way as robust polynomials: 
its output is hardly affected by small errors on its input bits.
We now want to derive a robust polynomial from this robust algorithm.
However, Corollary~\ref{cor2Tpoly} only deals with algorithms acting on the usual non-noisy type of oracles.
We circumvent this problem as follows. Pick a sufficiently large integer $m$, and fix error-fractions $\eps_i\in[0,\eps]$ 
that are multiples of $1/m$. Convert an input $x\in{\bitset}^n$ into $X\in{\bitset}^{nm}=X_1\ldots X_n$, 
where each $X_i$ is $m$ copies of $x_i$ but with an $\eps_i$-fraction of errors (the errors can be placed arbitrarily among the $m$ copies of $x_i$).
Note that the following map is an $\eps$-bounded-error oracle for $x$ that can be implemented by one query to $X$:
$$
\ket{i,b,0}\mapsto\ket{i}\frac{1}{\sqrt{m}}\sum_{j=1}^m\ket{b\oplus X_{ij}}\ket{j}=\sqrt{1-\eps_i}\,\ket{i,b\oplus x_i,w_i}+\sqrt{\eps_i}\,\ket{i,\overline{b\oplus x_i},w'_i}\,.
$$
Now consider the algorithm that Theorem~\ref{throbustalgo} provides for this oracle.
This algorithm makes $O(n)$ queries to~$X$, it is independent of the specific values of $\eps_i$
or the way the errors are distributed over $X_i$, and it has success probability $\geq 1-\eps$ as long as $\eps_i\leq\eps$ for each $i\in[n]$.
Applying Corollary~\ref{cor2Tpoly} to this algorithm gives an $nm$-variate multilinear polynomial $p$ in $X$ of degree $d=O(n)$.
This $p(X)$ lies in $[0,1]$ for every input $X\in{\bitset}^{nm}$ (since it is a success probability),
and has the property that $p(X_1,\ldots,X_n)$ is $\eps$-close to $f(x_1,\ldots,x_n)$ whenever $|X_i|/m$ is $\eps$-close to $x_i$ for each $i$.

It remains to turn each block $X_i$ of $m$ Boolean variables into one real-valued variable $z_i$.
This can be done by the method of \emph{symmetrization}~\cite{minsky&papert:perceptrons} as follows.
Define a new polynomial $p_1$ which averages $p$ over all permutations of the $m$ bits in $X_1$:
$$
p_1(X_1,\ldots,X_n)=\frac{1}{m!}\sum_{\pi\in S_m} p(\pi(X_1),X_2,\ldots,X_m)\,.
$$
Symmetrization replaces terms like $X_{11}\cdots X_{1t}$ by 
$$
V_t(X_1)=\frac{1}{\binom{m}{t}}\sum_{T\in\binom{[m]}{t}}\prod_{j\in
  T}X_{1j}\,.
$$
Therefore $p_1$ will be a linear combination of terms of the form 
$V_t(X_1)r(X_2,\ldots,X_n)$ for $t\leq d-\deg(r)$.
On $X_1\in{\bitset}^m$ of Hamming weight $|X_1|$, the sum $V_t(X_1)$
evaluates to 
$$
\binom{|X_1|}{t}=\frac{|X_1|(|X_1|-1)\cdots(|X_1|-t+1)}{t!}\,,
$$ 
which is a polynomial in $|X_1|=\sum_{j=1}^m X_{1j}$ of degree $t$.
Hence we can define $z_1=|X_1|/m$, and replace $p_1$ by a polynomial $q_1$ of total degree at most $d$ in $z_1,X_2,\ldots,X_m$,
such that $p_1(X_1,\ldots,X_n)=q_1(|X_1|/m,X_2\ldots,X_n)$.
We thus succeeded in replacing the block $X_1$ by one real variable $z_1$.
Repeating this for $X_2,\ldots,X_n$, we end up with a polynomial $q(z_1,\ldots,z_n)$ such
that $p(X_1,\ldots,X_n)=q(|X_1|/m,\ldots,|X_n|/m)$ for all $X_1,\ldots,X_n\in{\bitset}^{nm}$.
This $q$ will not be multilinear anymore, but it has degree at most $d=O(n)$ and it $\eps$-robustly approximates $f$:
for every $x\in{\bitset}^n$ and for every $z\in[0,1]^n$ satisfying $|z_i-x_i|\leq\eps$ for all $i\in[n]$, we have that $q(z)$ and $f(x)$ are $\eps$-close.
(Strictly speaking we have only dealt with the case where the $z_i$ are multiples of $1/m$, but we can choose $m$ as large as we want
and a low-degree polynomial cannot change much if its input varies between $i/m$ and $(i+1)/m$.)

\begin{corollary}[BNRW]
For every Boolean function $f$, there exists an $n$-variate polynomial of degree $O(n)$ that $\eps$-robustly approximates $f$.
\end{corollary}

\subsection{Closure properties of $\PP$} \label{sec:postbqp}

The important classical complexity class $\PP$ consists of all languages $L$ for which there exists a probabilistic polynomial-time algorithm that accepts an input $x$ with probability at least $1/2$ if $x \in L$, and with probability less than $1/2$ if $x \notin L$.  Note that under this criterion, the algorithm's acceptance probabilities may be extremely close to $1/2$, so $\PP$ is not a realistic definition of the class of languages feasibly computable with classical randomness.  Indeed, it is not hard to see that $\PP$ contains $\NP$.
Still, $\PP$ is worthy of study because of its many relations to other complexity classes.

One of the most basic questions about a complexity class $\mathcal{C}$ is which \emph{closure properties} it possesses.  For example, if $L_1, L_2 \in \mathcal{C}$, is $L_1 \cap L_2 \in \mathcal{C}$?  That is, is $\mathcal{C}$ \emph{closed under intersection}?  In the case of $\PP$, this question was posed by Gill~\cite{gill:pp}, who defined the class, and was open for many years before being answered affirmatively by Beigel et al.~\cite{brs:ppclosed}.  It is now known that $\PP$ is closed under significantly more general operations~\cite{brs:ppclosed, fortnow&reingold:pptt, aaronson:pp}.  Aaronson~\cite{aaronson:pp} gave a new and arguably more intuitive proof of the known closure properties of $\PP$, by providing a \emph{quantum} characterization of $\PP$.

To describe this result, we first briefly introduce the model of
quantum polynomial-time computation.  A \emph{quantum circuit} is a
sequence of unitary operations $U_1, \ldots, U_T$, applied to the
initial state $\ket{x}\ket{0^m}$, where $x \in \{0, 1\}^n$ is the
input to the circuit and $\ket{0^m}$ is an auxiliary workspace.  By
analogy with classical circuits, we require that each $U_t$ be a
\emph{local} operation which acts on a constant number of qubits.  For
concreteness, we require each $U_t$ to be a Hadamard gate, or the
single-qubit operation 
$$
\left( \begin{array}{cc}
1 & 0 \\
0 & e^{\frac{i\pi}{4}} \end{array} \right)\,,
$$
or the two-qubit controlled-NOT gate (which maps computational basis states $\ket{a,b}\mapsto\ket{a,a\oplus b}$).  A computation ends by measuring the first workspace qubit.  We say that such a circuit computes a function $f_n:{\bitset}^n\rightarrow{\bitset}$ \emph{with bounded error} if on each $x \in {\bitset}^n$, the final measurement equals $f_n(x)$ with probability at least $2/3$.  $\BQP$ is the class of languages computable with bounded error by a logspace-uniform family of polynomial-size quantum circuits.  Here, both the workspace size and the number of unitaries are required to be polynomial.  The collection of gates we have chosen is \emph{universal}, in the sense that it can efficiently simulate any other collection of local unitaries to within any desired precision~\cite[Section~4.5.3]{nielsen&chuang:qc}.  Thus our definition of $\BQP$ is a robust one.

In~\cite{aaronson:pp}, Aaronson investigated the power of a
``fantasy'' extension of quantum computing in which an algorithm may
specify a desired outcome of a measurement in the standard basis, and
then condition the quantum state upon seeing that outcome (we require
that this event have nonzero probability).  Formally, if a quantum
algorithm is in the pure state $\ket{\psi} = \ket{\psi_0}\ket{0} +
\ket{\psi_1}\ket{1}$ (where we have distinguished a 1-qubit register
of interest, and $\ket{\psi_1}$ is non-zero), then the
\emph{postselection} transformation carries $\ket{\psi}$ to 
$$
\frac{\ket{\psi_1}\ket{1}}{\sqrt{\braket{\psi_1}{\psi_1}}}\, .
$$
The complexity class $\PostBQP$ is defined as the class of languages computable with bounded error by a logspace-uniform family of polynomial-size quantum circuits that are allowed to contain postselection gates.  We have:

\begin{theorem}[Aaronson]\label{thpostbqp} 
$\PP= \PostBQP$.
\end{theorem}

From Theorem~\ref{thpostbqp}, the known closure properties of $\PP$ follow easily.  For example, it is clear that if $L_1, L_2 \in \PostBQP$, then we may amplify the success probabilities in the $\PostBQP$ algorithms for these languages, then simulate them and take their AND to get a $\PostBQP$ algorithm for $L_1 \cap L_2$.
This shows that $\PostBQP$ (and hence also $\PP$) is closed under intersection.

\medskip

\begin{proof}[Proof (sketch)] 
  We begin with a useful claim about postselection: any quantum
  algorithm with postselection can be modified to make just a single
  postselection step after all its unitary transformations (but before
  its final measurement).  We say that such a postselection algorithm
  is in \emph{canonical form}.  To achieve this, given any $\PostBQP$
  algorithm $A$ for a language $L$, consider a new algorithm $A'$
  which on input $x$, simulates $A(x)$.  Each time $A$ makes a
  postselecting measurement on a qubit, $A'$ instead records that
  qubit's value into a fresh auxiliary qubit.  At the end of the
  simulation, $A'$ postselects on the event that all these recorded
  values are 1, by computing their AND in a final auxiliary qubit
  $\ket{z}$ and postselecting on $\ket{z} = \ket{1}$.  The final state
  of $A'(x)$ is equivalent to the final state of $A(x)$, so $A'$ is a
  $\PostBQP$ algorithm for $L$ and is in canonical form.  This conversion
  makes it easy to show that $\PostBQP \subseteq \PP$, by the same
  techniques which show $\BQP \subseteq
  \PP$~\cite{adh:qcomputability}.  We omit the details, and turn to
  show $\PP \subseteq \PostBQP$.

Let $L \in \PP$, and let $M$ be a probabilistic polynomial-time algorithm witnessing this fact.  Say $M$ uses $m = \poly(n)$ random bits on an input $x$ of length $n$.  Then any such input defines a function
$
g = g_x : \{0, 1\}^m \rightarrow \{0, 1\},
$
by the rule
$$
g(r) = [\text{$M(x)$ accepts when using $r$ as its random string}]\,.
$$
By definition, $x \in L \Leftrightarrow |g^{-1}(1)| \geq 2^{m - 1}$.  We show how to determine whether $|g^{-1}(1)| \geq 2^{m - 1}$, in quantum polynomial time with postselection.  Let $s = |g^{-1}(1)|$; we may assume without loss of generality that $s > 0$.  The core of the algorithm is the following subroutine, which uses postselection to produce a useful quantum state:
\begin{enumerate}
\item[A1.] Initialize an $(m + 1)$-bit register to $\ket{0^{m+1}}$.
  Apply $H^{\otimes m}$ to the first $m$ qubits, then apply $g$ to
  these qubits and add the result into the $(m+1)$-st qubit, yielding
  the state 
$$
\frac{1}{\sqrt{2^m}}\sum_{x \in {\bitset}^m }\ket{x}\ket{g(x)}\,.
$$
\item[A2.]  Apply a Hadamard gate to each qubit of the $x$-register, yielding
$$
 \frac{1}{\sqrt{2^m}}\sum_{x \in {\bitset}^m} \left( \frac{1}{\sqrt{2^m}}\sum_{w \in {\bitset}^m} (-1)^{w\cdot x}\ket{w}     \right) \ket{g(x)}\,,
$$
where $w\cdot x=\sum_i w_ix_i$ denotes inner product of bit strings.  
Note that $0^m \cdot x = 0$ for all $x$, so the component of the above
state with first register equal to $0^m$ is 
$$
\frac{1}{2^m} \left( (2^m - s)\ket{0} + s \ket{1} \right)\,.
$$
\item[A3.] Postselect on the first $m$ qubits measuring to $0^m$.  This yields the reduced state
$$
\ket{\psi} = \frac{(2^m -s)\ket{0} + s \ket{1}}{\sqrt{ (2^m - s)^2  + s^2 }}
$$
on the last qubit.
\end{enumerate}

Using this subroutine, we can prepare fresh copies of $\ket{\psi}$ on demand.  Now let $(\alpha, \beta)$ be a pair of positive reals  to be specified later, satisfying $\alpha^2 + \beta^2 = 1$.  Using this pair, we define a second subroutine:
\begin{enumerate}
\item[B1.] Prepare a qubit $\ket{z} = \alpha\ket{0} + \beta\ket{1}$.  Perform a controlled-Hadamard operation on a fresh copy of $\ket{\psi}$ with control qubit $\ket{z}$, yielding the joint state $\alpha \ket{0}\ket{\psi} + \beta \ket{1}H\ket{\psi}$.  Note that
$$
H\ket{\psi} = \frac{\frac{1}{\sqrt{2}}\left(2^m\ket{0}  +  (2^m - 2s)\ket{1} \right) }{\sqrt{ (2^m - s)^2 + s^2  }}\,.
$$
\item[B2.] Now postselect on the second qubit measuring to 1, yielding the reduced state
$$
\ket{z'} = \frac{  \alpha s \ket{0}  + \beta \frac{1}{\sqrt{2}} (2^m - 2s)\ket{1}  }{\sqrt{   \alpha^2s^2 + (\beta^2/2)(2^m - 2s)^2    }}\,.
$$
Perform a projective measurement relative to the basis 
$$
\{\ket{+}, \ket{-}\} = \left\{\frac{\ket{0} + \ket{1}}{\sqrt{2}},
\frac{\ket{0} - \ket{1}}{\sqrt{2}}\right\}\,,
$$
recording the result.
\end{enumerate}
Finally, our algorithm to determine if $s \geq 2^{m - 1}$ is as follows: we perform steps B1-B2 for each choice of $(\alpha, \beta)$ from a collection $\{(\alpha_i, \beta_i)\}_{i = -m}^m$ of pairs chosen to satisfy ${\alpha_i}/{\beta_i}= 2^i$, $i \in \{-m, \ldots,  m\}$.  For each choice of $i$, we repeat steps B1-B2 for $O(\log m)$ trials, and use the results to estimate the probability $|\inp{+}{z'}|^2$ of outcome $\ket{+}$ under $(\alpha_i, \beta_i)$.

The idea of the analysis is that, if $s < 2^{m - 1}$, then $\ket{z'}$ is in the strictly-positive quadrant of the real plane with axes $\ket{0}, \ket{1}$; here we are using that $\alpha, \beta$ are positive.  Then a suitable choice of $(\alpha_i, \beta_i)$ will cause $\ket{z'}$ to be closely aligned with $\ket{+}$, and will yield measurement outcome $\ket{+}$ with probability close to $1$.  On the other hand, if $s \geq 2^{m - 1}$, then $\ket{z'}$ is never in the same quadrant as $\ket{+}$ or the equivalent state $-\ket{+}$.  Then, for any choice of $(\alpha, \beta)$ we have $|\braket{+}{z'}| \leq 1/\sqrt{2}$ and the probability of measuring $\ket{+}$ is at most $1/2$.  Thus our repeated trials allow us to determine with high probability whether or not $s \geq 2^{m - 1}$.  We conclude that $\PP \subseteq \PostBQP$.  
\end{proof}

The original proof that $\PP$ is closed under
intersection~\cite{brs:ppclosed} crucially relied on ideas from
approximation by rational functions, i.e., by ratios of polynomials.
Aaronson's proof can also be seen as giving a new method for building
rational approximations.  To see this, note that the postselection
operation also makes sense in the quantum query model, and that in
this model as well algorithms may be put in canonical form (with a
single postselection).  Suppose a canonical-form query algorithm with
postselection makes $T$ quantum queries.  If its state before the
postselection step is $\ket{\psi^T} = \ket{\psi^T_0}\ket{0} +
\ket{\psi^T_1}\ket{1}$, then the amplitudes of $\ket{\psi^T_1}$ are
degree-$T$ multilinear polynomials in the input $x$, by
Corollary~\ref{cor2Tpoly}.  Then inspecting the postselected state 
$$
\frac{\ket{\psi^T_1}\ket{1}}{\sqrt{\braket{\psi^T_1}{\psi^T_1}}}\,,
$$
we see that the \emph{squared} amplitudes are rational functions of $x$ of degree $2T$ in $x$, that is, each squared amplitude can be expressed as $p(x)/q(x)$, where both numerator and denominator are polynomials of degree at most $2T$.  Moreover, all squared amplitudes have the same denominator $q(x)=\braket{\psi^T_1}{\psi^T_1}$. %
Thus in the final decision measurement, the acceptance probability is a degree-$2T$ rational function of $x$.

This may be a useful framework for designing other rational approximations.  In fact, one can show the quantum approach to rational approximation always gives an essentially optimal degree: if $f$ is $1/3$-approximated by a rational function of degree $d$, then $f$ is computed by a quantum query algorithm with postselection, making $O(d)$ queries~\cite{wolf:ratdegnote}, and vice versa.  
Such a tight, constant-factor relation between query complexity and degree is known to be false for the case of the usual (non-postselected)
quantum query algorithms and the usual approximate degree~\cite{ambainis:degreevsqueryj}.

\subsection{Jackson's theorem}

In this section we describe how a classic result in approximation theory, Jackson's Theorem, can be proved using quantum ideas.

The approximation of complicated functions by simpler ones, such as polynomials, has interested mathematicians for centuries.  A fundamental result of this type is Weierstrass's Theorem from 1885~\cite{weierstrass}, which states that any \emph{continuous} function $g: [0, 1] \rightarrow \mathbb{R}$ can be arbitrarily well-approximated by polynomials.  For any function $f: [0, 1] \rightarrow \mathbb{R}$, let $\norm{f}_\infty = \sup_{x \in [0, 1]}|f(x)|$; then Weierstrass's Theorem states that for any continuous $g$ and any $\eps > 0$, there exists a polynomial $p(x)$ such that $\norm{p - g}_\infty < \eps$.

In 1912 Bernstein gave a simple construction of such polynomials, which can be described in a \emph{probabilistic} fashion~\cite{bernstein:weierstrass,alon&spencer:probmethod}.  Fix a continuous function $g$ on $[0, 1]$ and a value $n \geq 1$, and consider the following algorithm: flip $n$ coins, each taking value 1 with probability $x$.  
Let $t\in\{0,\ldots,n\}$ be the number of 1s observed in the resulting string $X\in{\bitset}^n$. Output $g(t/n)$. 
Note that the expected value of the random variable $t/n$ is exactly $x$, and its standard deviation is $O(1/\sqrt{n})$.
Consider the \textit{expected value} of the output of this procedure, given a bias $x$:
$$
B_{g, n}(x)=\Exp\left[g\left(|X|/n\right)\right]=\sum_{t=0}^n\Pr\left[|X|=t\right]\cdot g\left(t/n \right).
$$ 
We have 
$$
\Pr\left[|X|=t\right]=\binom{n}{t}x^t(1-x)^{n-t}\,,
$$
hence $B_{g, n}(x)$ is a polynomial in $x$ of degree at most $n$.  Moreover, $B_{g, n}(x)$ is intuitively a good estimator for $g(x)$ since $t/n$ is usually close to $x$.  To quantify this, let $\omega_{\delta}(g)$, the \textit{modulus of continuity of $g$ at scale} $\delta$, be the supremum of $|g(x) - g(y)|$ over all $x,y\in [0, 1]$ such that $|x - y| \leq \delta$.  Then it can be shown~\cite{Riv} that $\norm{B_{g, n} - g}_\infty = O(\omega_{ 1/\sqrt{n} }(g))$.  This is not too surprising since we expect the fraction of heads observed to be concentrated in an interval of length $O(1/\sqrt{n})$ around $x$.  The interval $[0, 1]$ is compact, so $g$ is \emph{uniformly} continuous and $\omega_{\delta}(g) \rightarrow 0$ as $\delta \rightarrow 0$.  Thus the Bernstein polynomials yield a quantitative refinement of Weierstrass's Theorem.

Around the same time as Bernstein's work, an improved result was shown by Jackson~\cite{jackson:thesis}:

\begin{theorem}[Jackson]\label{thmjackson} 
If $g$ is continuous, then for all $n \geq 1$ there exists a polynomial $J_{g, n}$ of degree $n$ such that $\norm{J_{g, n} - g}_\infty = O(\omega_{1/n}(g))$.  
\end{theorem}

In Jackson's theorem the quality of approximation is based on the maximum fluctuation of $g$ at a much smaller scale than Bernstein's
($1/n$ instead of $1/\sqrt{n}$).  Up to the constant factor, Jackson's Theorem is optimal for approximation guarantees based on the modulus of continuity.  
The original proof used trigonometric ideas.
Drucker and de Wolf~\cite{drucker&wolf:jackson} gave a proof of Jackson's Theorem that closely follows Bernstein's idea, 
but replaces his classical estimation procedure with the quantum counting algorithm of~\cite{bhmt:countingj}.
As mentioned in Section~\ref{ssecqueryalgos}, with $M$ queries this algorithm produces an estimate $\tilde{t}$ of the Hamming weight $t=|X|$ of a string $X\in{\bitset}^N$, such that for every integer $j\geq 1$ we have $\Pr[|\tilde{t} - t| \geq jN/M] = O(1/j)$. For our purposes we need an even sharper estimate.  To achieve this, let $Count'(X, M)$ execute $Count(X, M)$ five times, yielding estimates $\tilde{t}_1,\ldots,\tilde{t}_5$, and output their median, denoted $\tilde{t}_{med}$.\footnote{A more careful analysis in~\cite{drucker&wolf:jackson} allows a median of three trials to be used.}  Note that for $|\tilde{t}_{med} - t | \geq j N/M$ to hold, we must have $|\tilde{t}_{i} - t | \geq jN/M$ for at least three of the trials $i \in \{1,\ldots,5\}$. This happens with probability at most $O(1/j^3)$, which implies 
\begin{equation}\label{eqtmedbound}
\Exp\left[|\tilde{t}_{med} - t |\right] \leq \sum_{j\geq 1} \frac{jN}{M} \cdot O(1/j^3)=O\left(N/M\right).
\end{equation}
With this estimator in hand, let us fix a continuous function $g$ and $n \geq 1$, and consider the following algorithm $A(x)$ to estimate $g(x)$ for an unknown value $x$: flip $N = n^2$ independent coins, each taking value 1 with probability $x$, yielding an $N$-bit string $X$.  Set $M = \lfloor n/10\rfloor$ and run $Count'(X, M)$, yielding an estimate $\tilde{t}_{med}$, and output $g(\tilde{t}_{med}/N)$.
This algorithm makes $5M \leq n/2$ queries to $X$, so
it follows from Corollary~\ref{cor2Tpoly} that its expected output value (viewed as a function of the $N$ bits of $X$) 
is a multilinear polynomial of degree at most $n$.  
Define $J_{g, n}(x) = \Exp[A(x)]$, where the expectation is taken both over the choice of $X$ and over the randomness in the output
of the quantum counting procedure.  
Note that $\Exp[X_{i_1}\ldots X_{i_d}]=x^d$, since the bits of $X$ are independent coin flips, each with expectation~$x$.
Hence $J_{g, n}$ is a degree-$n$ univariate polynomial in $x$.  
We bound $|J_{g, n}(x) - g(x)|$ as follows.  Consider the random variable $\tilde{x} = \tilde{t}_{med}/N$.  
Using the definition of $\omega_{1/n}(g)$, we have
\begin{eqnarray*}
|J_{g, n}(x) - g(x)| & = & | \Exp [g(\tilde{x})] -  g(x)|\\     
& \leq & \Exp [|g(\tilde{x}) - g(x)|]\\  
& \leq & \Exp[(n \cdot |\tilde{x} - x| + 1) \cdot \omega_{1/n}(g) ]\\
& = & \left(n \cdot \frac{1}{N}\Exp[|\tilde{t}_{med} - xN|] + 1\right) \cdot \omega_{1/n}(g)\\  
& \leq & \left(\frac{1}{n}(\Exp[|\tilde{t}_{med} - t |]  +  \Exp[| t - xN |] )       + 1\right) \cdot \omega_{1/n}(g)\,.
\end{eqnarray*}
Since $t=|X|$ has expectation $xN$ and variance $x(1-x)N$, we have $\Exp[| t - xN |] = O(\sqrt{N}) = O(n)$. 
By \eqref{eqtmedbound} we have $\Exp[|\tilde{t}_{med} - t |] = O(N/M) = O(n)$.  Plugging these two findings into our expression yields $|J_{g, n}(x) - g(x)|  = O(\omega_{1/n}(g))$, for all $x\in[0,1]$.  This proves Jackson's Theorem.

At the heart of quantum counting, one finds trigonometric ideas closely related to those used in some classical proofs of Jackson's Theorem (see~\cite{drucker&wolf:jackson} for a discussion).  Thus it is not really fair to call the above a simplified proof.  It does, however, show how Bernstein's elegant probabilistic approach to polynomial approximation can be carried further with the use of quantum algorithms.

\subsection{Separating strong and weak communication versions of $\PP$}\label{PPvsUPP}

The previous examples all used the connection between polynomials and quantum \emph{algorithms}, 
more precisely the fact that an efficient quantum (query) algorithm induces a low-degree polynomial.
The last example of this section uses a connection between polynomials and quantum \emph{communication protocols}.  
Communication complexity was introduced in Section~\ref{ssecip}.  Like computational complexity, communication complexity has a host of different models: 
one can study deterministic, bounded-error, or non-deterministic protocols, and so on.
Recall the complexity class $\PP$ from Section~\ref{sec:postbqp}.  This class has \emph{two} plausible analogues in the communication setting: an unbounded-error version and a weakly-unbounded-error version.
In the first, we care about the minimal $c$ such that there is a $c$-bit randomized communication protocol computing 
$f:{\bitset}^n\times{\bitset}^n\rightarrow{\bitset}$ correctly with success probability strictly greater than 1/2 on every input 
(though the success probability may be arbitrarily close to 1/2).
In the second version, we care about the minimal $c$ such that there is a $c$-bit randomized protocol computing 
$f$ correctly with success probability at least $1/2+\beta$ on every input, for $\beta\geq 2^{-c}$.
The motivation for the second version is that $c$ plays the role that computation-time takes for the computational class
$\PP$, and in that class we automatically have $\beta\geq 2^{-c}$: a time-$c$ Turing machine cannot flip more than $c$ fair coins, 
so all events have probabilities that are multiples of $1/2^c$.
Let $\UPP(f)$ and $\PP(f)$ correspond to the minimal communication $c$ in the unbounded and weakly-unbounded cases, respectively.
Both measures were introduced by Babai et al.~\cite{bfs:classes}, who asked whether they were approximately equal.  
(One direction is clear: $\UPP(f)\leq\PP(f)$).

These two complexity measures are interesting as models in their own right, but are all the more interesting
because they each correspond closely to fundamental complexity measures of Boolean matrices.
First, $\UPP(f)$ equals (up to a small additive constant) the log of the \emph{sign-rank} of $f$~\cite{paturi&simon:pcc}.  
This is the minimal rank among all $2^n\times 2^n$ matrices $M$ satisfying that the sign of the entry $M_{x,y}$ equals $(-1)^{f(x,y)}$ for all inputs.
Second, $\PP(f)$ is essentially the \emph{discrepancy} bound~\cite{klauck:qcclowerj}, 
which in turn is essentially the \emph{margin complexity} of $f$~\cite{linial&shraibman:ccj,linial&shraibman:learningj}.
The latter is an important and well-studied concept from computational learning theory.
Accordingly, an example where $\UPP(f)\ll \PP(f)$ also separates sign-rank on the one hand from discrepancy and margin complexity on the other.

Simultaneously but independently, Buhrman et al.~\cite{bvw:smallbias} and Sherstov~\cite{sherstov:halfspacej} (later improved in~\cite{sherstov:patternj})
found functions where $\UPP(f)$ is exponentially smaller than $\PP(f)$, answering the question from~\cite{bfs:classes}. 
We will give the proof of~\cite{bvw:smallbias} here.
Our quantum tool is the following lemma, first proved (implicitly) in~\cite{razborov:qdisj} and made more explicit in~\cite[Section~5]{ksw:dpt-siam}.

\begin{lemma}[Razborov]\label{lemrazborov}
Consider a quantum communication protocol on $m$-bit inputs $x$ and $y$
that communicates $q$ qubits, with outputs 0 and 1 (``reject'' and ``accept''), and acceptance probabilities denoted by $P(x,y)$.
For $i\in\{0,\ldots,m/4\}$, define
$$
P(i)=\Exp_{|x|=|y|=m/4, |x\wedge y|=i}[P(x,y)]\,,
$$
where the expectation is taken uniformly over all $x,y\in{\bitset}^m$
that each have Hamming weight $m/4$ and that have intersection size $i$.
For every $d\leq m/4$ there exists a univariate degree-$d$ polynomial $p$ (over the reals)
such that $|P(i)-p(i)|\leq 2^{-d/4+2q}$ for all $i\in\{0,\ldots,m/8\}$.
\end{lemma}

If we choose degree $d=8q+4\log(1/\eps)$, then $p$ approximates $P$ to within an additive $\eps$ for all $i\in\{0,\ldots,m/8\}$.
This allows us to translate quantum protocols to polynomials.
The connection is not as clean as in the case of quantum query algorithms (Corollary~\ref{cor2Tpoly}): 
the relation between the degree and the quantum communication complexity is less tight, and the resulting polynomial 
only \emph{approximates} certain average acceptance probabilities instead of exactly representing 
the acceptance probability on each input.

Lemma~\ref{lemrazborov} allows us to prove lower bounds on quantum communication complexity by
proving lower bounds on polynomial degree. The following result~\cite{ehlich&zeller:schwankung,rivlin&cheney:approx}
is a very useful degree bound:

\begin{theorem}[Ehlich and Zeller, Rivlin and Cheney]\label{thehlichzeller}
Let $p: \R\rightarrow\R$ be a polynomial such that
$b_1\leq p(i)\leq b_2$ for every integer $0\leq i\leq n$, and
its derivative has $|p'(x)|\geq c$ for some real $0\leq x\leq n$.
Then $\deg(p)\geq\sqrt{cn/(c+b_2-b_1)}$.
\end{theorem}

\paragraph{The function and its unbounded-error protocol}
To obtain the promised separation, we use a distributed version of the ODD-MAX-BIT function of Beigel~\cite{beigel:perceptrons}.
Let $x,y\in{\bitset}^n$, and $k=\max\{i\in[n] : x_i=y_i=1\}$ be the rightmost position where $x$ and $y$ both have a 1
(set $k=0$ if there is no such position).
Define $f(x,y)$ to be the least significant bit of $k$, i.e., whether this $k$ is odd or even.
Buhrman et al.~\cite{bvw:smallbias} proved:

\begin{theorem}[BVW]
For the distributed ODD-MAX-BIT function we have $\UPP(f)=O(\log n)$ and $\PP(f)=\Omega(n^{1/3})$.
\end{theorem}

The upper bound is easy:
For $i\in[n]$, define probabilities $p_i=c2^i$, where $c=1/\sum_{i=1}^n 2^i$ is a normalizing constant.
Consider the following protocol.
Alice picks a number $i\in[n]$ with probability $p_i$ and sends over $(i,x_i)$ using $\ceil{\log n} + 1$ bits.
If $x_i=y_i=1$ then Bob outputs the least significant bit of $i$, otherwise he outputs a fair coin flip.
This computes $f$ with positive (but exponentially small) bias.
Hence $\UPP(f)\leq\ceil{\log n} + 1$.

\paragraph{Weakly-unbounded-error lower bound}
It will be convenient to lower-bound the complexity of \emph{quantum} protocols computing $f$ with weakly-unbounded error,
since this allows us to use Lemma~\ref{lemrazborov}.\footnote{We could just analyze classical protocols 
and use the special case of Razborov's result that applies to classical protocols.
However, the classical version of Razborov's lemma was not known prior to~\cite{razborov:qdisj},
and arguably would not have been discovered if it were not for the more general quantum version.
We would end up with the same communication lower bound anyway, since quantum and classical weakly-unbounded 
error complexity turn out to be essentially the same~\cite{klauck:qcclowerj,ipry:qpp}.}
Consider a quantum protocol with $q$ qubits of communication that computes $f$ with bias $\beta>0$.
Let $\beta(x,y)=P(x,y)-1/2$. Then $\beta(x,y)\geq\beta$ if $f(x,y)=1$, and $\beta(x,y)\leq-\beta$ if $f(x,y)=0$.
Our goal is to lower-bound $q+\log(1/\beta)$.
We will do that below by a process that iteratively fixes some of the input bits, 
applying Lemma~\ref{lemrazborov} to each intermediate function to find a good fixing, 
in a way that produces a bounded function with larger and larger ``swings.''

Define $d=\ceil{8q+4\log(2/\beta)}$ and $m=32d^2+1$.
Assume for simplicity that $2m$ divides $n$.
We partition $[n]$ into $n/2m$ consecutive intervals (or ``blocks''), each of length $2m$.
In the first interval (from the left), fix the bits $x_i$ and $y_i$ to 0 for even $i$;
in the second, fix $x_i$ and $y_i$ to 0 for odd $i$;
in the third, fix $x_i$ and $y_i$ to 0 for even $i$, etc.
In the $j$-th interval there are $m$ unfixed positions left for each party.
Let $x^{(j)}$ and $y^{(j)}$ denote the corresponding $m$-bit strings in $x$ and $y$, respectively.

We will define inductively, for all $j=1,2,\dots,n/2m$, particular
strings $x^{(j)}$ and $y^{(j)}$ as follows.
Let $X^{j}$ and $Y^{j}$ denote $n$-bit strings where
the first $j$ blocks are set to $x^{(1)},\dots, x^{(j)}$ and
$y^{(1)},\dots, y^{(j)}$, respectively, and all the other blocks are set to $0^{2m}$.
In particular, $X^{0}$ and $Y^{0}$ are all zeros.
We will define $x^{(j)}$ and $y^{(j)}$ so that
$$
\beta(X^j,Y^j)\ge 2^j\beta\qquad  \text{ or }\qquad
\beta(X^j,Y^j)\le -2^j\beta
$$
depending on whether $j$ is odd or even.
This holds automatically for $j=0$, which is the base case of the induction.

Now assume $x^{(1)},\dots,x^{(j-1)}$
and $y^{(1)},\dots,y^{(j-1)}$ are already defined on previous steps.
On the current step, we have to define $x^{(j)}$ and $y^{(j)}$.
Without loss of generality assume that $j$ is odd, thus we have
$\beta(X^{j-1},Y^{j-1})\le -2^{j-1}\beta$.
Consider each $i=0,1,\dots,m/4$.
Run the protocol on the following distribution:
$x^{(j)}$ and $y^{(j)}$ are chosen randomly subject to each having Hamming weight $m/4$ and having
intersection size $i$; the blocks with indices smaller than $j$ are
fixed (on previous steps), and the blocks with indices larger than $j$ are set to zero.
Let $P(i)$ denote the expected value (over this distribution) of $\beta(x,y)$ as a function of $i$.
Note that for $i=0$ we have $P(i)=\beta(X^{j-1},Y^{j-1})\le -2^{j-1}\beta$.
On the other hand, for each $i>0$ the expectation is taken over $x,y$ with $f(x,y)=1$,
because the rightmost intersecting point is in the $j$-th interval and hence odd
(the even indices in the $j$-th interval have all been fixed to 0).
Thus $P(i)\geq\beta$ for those $i$.
Now assume, by way of contradiction, that $\beta(X^j,Y^j)<2^{j}\beta$
for all $x^{(j)}, y^{(j)}$ and hence $P(i)<2^{j}\beta$ for all such $i$.
By Lemma~\ref{lemrazborov}, for our choice of $d$, we can approximate $P(i)$ to within additive error
of $\beta/2$ by a polynomial $p$ of degree $d$.
(Apply Lemma~\ref{lemrazborov} to the protocol obtained from the original
protocol by fixing all bits outside the $j$-th block.)
Let $r$ be the degree-$d$ polynomial
$$
r = \frac{p-\beta/2}{2^{j-1}\beta}.
$$
From the properties of $P$ and the fact that $p$ approximates $P$ up to $\beta/2$,
we see that $r(0)\leq -1$ and $r(i)\in[0,2]$ for all $i\in[m/8]$.
But then by the degree lower bound of Theorem~\ref{thehlichzeller}, 
$$
\deg(r)\geq\sqrt{(m/8)/4}=\sqrt{d^2+1/32}>d\,,
$$
which is a contradiction.
Hence there exists an intersection size $i\in[m/8]$ where $P(i)\geq 2^{j}\beta$.
Thus there are particular $x^{(j)},y^{(j)}$ with $\beta(X^j,Y^j)\ge 2^{j}\beta$, concluding the induction step.

For $j=n/2m$ we obtain $|\beta(X^j,Y^j)|\ge 2^{n/2m}\beta$.
But for every $x,y$ we have $|\beta(x,y)|\leq 1/2$ because $1/2+\beta(x,y)=P(x,y)\in[0,1]$, 
hence $1/2\geq 2^{n/2m}\beta$. This implies $2m\log(1/2\beta)\geq n$, and therefore
$$
(q+\log(1/\beta))^3  \geq  (q+\log(1/\beta))^2 \log(1/\beta)
                      =    \Omega(m \log(1/\beta))
                      =    \Omega(n)\,.
$$
This allows us to conclude $\PP(f)=\Omega(n^{1/3})$, which is exponentially larger than $\UPP(f)$.

\section{Other applications}\label{secotherapps}

In this section we go over a few further examples of the use of quantum techniques in non-quantum results,
examples that do not fit in the two broad categories of the previous two sections.
First we give two examples of classical results that were both \emph{inspired} by earlier quantum proofs, 
but do not explicitly use quantum techniques.%
\footnote{This is reminiscent of a famous metaphor in Ludwig Wittgenstein's \emph{Tractatus logico-philosophicus} (6.45), 
about the ladder one discards after having used it to climb to a higher level.}
Finally, in Section~\ref{ssecfurtherlit} we give a brief guide to further literature with more examples.

\subsection{The relational adversary}\label{secadversary}

In Section~\ref{ssecqueryalgos} we described the \emph{polynomial method} for proving lower bounds for quantum query complexity.  This method has a strength which is also a weakness: it applies even to a stronger (and less physically meaningful) model of computation where we allow \emph{any linear transformation} on the state space, not just unitary ones.  As a result, it does not always provide the strongest possible lower bound for quantum query algorithms.

Ambainis~\cite{ambainis:lowerboundsj, ambainis:degreevsqueryj}, building on an earlier method of Bennett et al. \cite{bbbv:str&weak}, addressed this problem by providing a general method for quantum lower bounds, the \emph{quantum adversary}, which exploits unitarity in a crucial way and which in certain cases yields a provably better bound than the polynomial method~\cite{ambainis:degreevsqueryj}.  Surprisingly, Aaronson~\cite{aaronson:localsearchj} was able to modify Ambainis's argument to obtain a new method, the \emph{relational adversary},  for proving \textit{classical} randomized query lower-bounds.  He used this method to give improved lower bounds on the complexity of the ``local search'' problem, in which one is given a real-valued function $F$ defined on the vertices of a graph $G$, and must locate a local minimum of $F$.  In this section we state Ambainis's lower-bound criterion, outline its proof, and describe how Aaronson's method follows a similar outline.\footnote{In both cases we state a simplified (``unweighted'') version of the lower-bound method in question, which conveys the essence of the technique.
After Ambainis's original paper, a version of the lower bound was formulated that allows even \emph{negative} weights~\cite{hls:madv}. 
Reichardt~\cite{reichardt:tight} recently proved that this general adversary bound in fact characterizes quantum query complexity.}

Recall from Section~\ref{ssecqueryalgos} that a quantum query algorithm is a sequence of unitary operations 
$$
U_TO_xU_{T-1}O_x\cdots O_xU_1O_xU_0\,,
$$
applied to the fixed starting state $\ket{0\ldots 0}$, where the basic ``query transformation'' $O_x$ depends on the input $x$ and $U_0,U_1, \ldots, U_T$ are arbitrary unitaries.  Ambainis invites us to look simultaneously at the evolution of our quantum state under all possible choices of $x$; formally, we let $\ket{\psi^t_x}$ denote the state at time $t$ (i.e., after applying $O_x$ for the $t$-th time) under input $x$.  In particular, $\ket{\psi^0_x} = \ket{0\ldots 0}$ for all $x$ (and $\braket{\psi^0_x}{\psi^0_y} = 1$ for each $x, y$).  Now if the algorithm computes the Boolean function $f$ with success probability $2/3$ on every input, then the final measurement must accept any pair $x \in f^{-1}(0), y \in f^{-1}(1)$ with success probabilities differing by at least $1/3$.  It is not hard to verify that this implies $|\braket{\psi^T_x}{\psi^T_y}| \leq {17}/{18}$.\footnote{Contrapositively, and in greater generality, if $|\braket{\psi_1}{\psi_2}| \geq 1 - \eps$ then under any measurement, $\ket{\psi_1}$ and $\ket{\psi_2}$ have acceptance probabilities differing by at most $\sqrt{2\eps}$ (see~\cite[Section~9.2]{nielsen&chuang:qc}).} This suggests that, for any given relation $R \subseteq f^{-1}(0) \times f^{-1}(1)$ (which one may regard as a bipartite graph), we consider the \emph{progress measure}
$$
S_t = \sum_{(x, y) \in R}|\braket{\psi^t_x}{\psi^t_y}|
$$
as a function of $t$.  By our observations, initially we have $S_0 = |R|$, and in the end we must have $S_T \leq ({17}/{18})\cdot|R|$.  Also, crucially, the progress measure is \emph{unaffected} by each application of a unitary $U_t$, since each $U_t$ is independent of the input and unitary transformations preserve inner products.

If we can determine an upper-bound $\Delta$ on the change $|S_{t + 1} - S_t|$ in the progress measure at each step, we can conclude that the number $T$ of queries is at least ${|R|}/{18\Delta}$.  Ambainis~\cite{ambainis:lowerboundsj, ambainis:degreevsqueryj} provides a condition on $R$ that allows us to derive such a bound:

\begin{theorem}[Ambainis]\label{qadv}  Suppose that the relation $R$ satisfies
\begin{itemize}
\item[(i)] each $x \in f^{-1}(0)$ appearing in $R$ appears at least $m_0$ times in $R$;
\item[(ii)] each $y \in f^{-1}(1)$ appearing in $R$ appears at least $m_1$ times in $R$;
\item[(iii)] for each $x \in f^{-1}(0)$ and $i \in [n]$, there are at most $\elll_0$ inputs $y \in f^{-1}(1)$ such that $(x, y) \in R$ and $x_i\neq y_i$;
\item[(iv)] for each $y \in f^{-1}(1)$ and $i \in [n]$, there are at most $\elll_1$ inputs $x \in f^{-1}(0)$ such that $(x, y) \in R$ and $x_i\neq y_i$.
\end{itemize}
Let the progress measure $S_t$ be defined relative to a fixed quantum algorithm as above. 
Then for all $t \geq 0$, 
$$
|S_{t + 1} - S_t| = O \left(\sqrt{   \frac{\elll_0}{m_0} \cdot
    \frac{\elll_1}{m_1}    } \cdot  |R|   \right)\,,
\quad\text{and therefore}\quad T = \Omega \left(\sqrt{
    \frac{m_0}{\elll_0} \cdot \frac{m_1}{\elll_1}    }   \right)\,.
$$
\end{theorem}

The art in applying Ambainis's technique lies in choosing the relation $R$ carefully to maximize this quantity.  Intuitively, conditions (i)-(iv) imply that $|S_{t + 1} - S_t|$ is small relative to $|R|$ by bounding the ``distinguishing ability'' of any query.

Every classical bounded-error algorithm is also a quantum bounded-error query algorithm, so the above lower-bound also applies in the classical case.  However, there are cases where this gives an inferior bound.  For example, for the promise problem of inverting a permutation, the above technique yields a query bound of $\Omega(\sqrt{n})$, which matches the true quantum complexity of the problem, while the classical randomized complexity is $\Omega(n)$.  In this and similar cases, the particular relation $R$ used in applying Theorem~\ref{qadv} for a quantum lower bound, is such that $\max \{ {m_0}/{\elll_0}, {m_1}/{\elll_1}  \}$ gives a good estimate of the \emph{classical} query complexity.  This led Aaronson to prove in~\cite{aaronson:localsearchj} a classical analogue of Ambainis's lower bound, in which the geometric mean of  ${m_0}/{\elll_0}$ and ${m_1}/{\elll_1}$ is indeed replaced with their maximum.

A sketch of Aaronson's proof is as follows.  Fixing the relation $R$, we use Yao's minimax principle to show that, if there is a randomized $T$-query bounded-error algorithm for computing $f$, then there is a \emph{deterministic} algorithm $A$ succeeding with high probability on a specific input distribution determined by $R$ (to be precise, pick a uniformly random pair $(x, y) \in R$ and select $x$ or $y$ with equal probability).  We now follow Ambainis, and consider for each input $x$ and $t \leq T$ the state $v_{t, x}$ (which is now a fixed, classical state) of the algorithm $A$ after $t$ steps on input $x$.  Let $I_{t, x, y}$ equal 1 if inputs $x$ and $y$ have not been distinguished by $A$ after $t$ steps, otherwise $I_{t, x, y} = 0$.  Define $S_t = \sum_{(x, y) \in R} I_{t, x, y}$ as our progress measure.\footnote{Actually, Aaronson~\cite{aaronson:localsearchj} defines an \emph{increasing} function,  but we modify this to show the similarity with Ambainis's proof.  We note that if the states $v_{t, x}$ are written as quantum amplitude vectors, and the state of the algorithm $A$ records all the information it sees, $I_{t, x, y}$ can actually be written as an inner product of quantum states just as in Ambainis's proof.} 

Similarly to the quantum adversary, we have $S_0 = |R|$, and Aaronson argues that the success condition of $A$ implies $S_T \leq (1 - \Omega(1))|R|$.  A combinatorial argument then yields the following result bounding the maximum possible change $|S_{t + 1} - S_t|$ after one (classical) query:

\begin{theorem}[Aaronson]\label{radv} Suppose that the relation $R$ obeys conditions (i)-(iv) in Theorem~\ref{qadv}.  Let the progress measure $S_t$ be defined relative to a deterministic algorithm $A$. 
Then for all $t \geq 0$, 
$$
|S_{t + 1} - S_t| = O \left( \min \left\{ \frac{\elll_0}{m_0},
  \frac{\elll_1}{m_1}    \right\}  \cdot |R|  \right)\quad\text{and
  therefore}\quad T = \Omega\left( \max \left\{ \frac{m_0}{\elll_0},
  \frac{m_1}{\elll_1}    \right\}   \right)\,.
$$
\end{theorem}

The details of Aaronson's proof are somewhat different from those of Theorem~\ref{qadv}, and there is no explicit use of quantum states, but the spirit is clearly similar, illustrating that the quantum query model is a sufficiently natural generalization of the classical model for ideas to flow in both directions.
Subsequently, Laplante and Magniez~\cite{laplante&magniez:lower} gave a different treatment of this based on Kolmogorov complexity, 
which brings out the analogy between the quantum and classical adversary bounds even more closely.

\subsection{Proof systems for the shortest vector problem}

A \textit{lattice} is an additive subgroup of $\mathbb{R}^n$ consisting of all integer combinations of a linearly independent set of $n$ vectors.  It can be shown that for every lattice $L$, there exists a value $\lambda(L) > 0$, the \emph{minimum (Euclidean) distance} of the lattice, such that: (i) any two distinct $x, y \in L$ are at distance at least $\lambda(L)$ from each other; (ii) there exists $x \in L$ such that $\norm{x} = \lambda(L)$.  Lattice vectors of length $\lambda(L)$ are called ``shortest (nonzero) vectors'' for $L$, and the problem of computing the minimum lattice distance is also known as the ``shortest vector problem'' (SVP).  

The problem of approximating $\lambda(L)$ to within a multiplicative factor $\gamma(n)$ can be equivalently formulated as a \emph{promise problem} called $\GapSVP_{\gamma(n)}$, in which we are given a basis for a lattice $L$, with the ``promise'' that either $\lambda(L) \leq 1$ or $\lambda(L) \geq \gamma(n)$, and we must determine which case holds.  A related problem is the ``closest vector problem'' (CVP), in which we are given a basis for a lattice $L$, and a ``target'' vector $v \notin L$, and would like to approximate the distance $d(v, L)$ from $v$ to the closest vector in the lattice.  In the promise problem $\GapCVP_{\gamma(n)}$, we are given $L$ and $v$ and must distinguish the case where $d(v, L) \leq 1$ from the case where $d(v, L) \geq \gamma(n)$.  It is known~\cite{GMSS99} that $\GapSVP_{\gamma(n)}$ reduces to $\GapCVP_{\gamma(n)}$ for any approximation factor $\gamma(n)$.

Approximate solutions to closest and shortest vector problems have numerous applications in pure and applied mathematics.  However, $\GapSVP_{\gamma(n)}$ has been proven intractable even for super-constant approximation factors $\gamma(n)$, albeit under an assumption somewhat stronger than $\Polytime \neq \NP$~\cite{khot:hardlatticej, HavivR07} (namely that $\NP$-complete problems are not solvable by a randomized algorithm in time $2^{\textrm{polylog}(n)}$.)  Even computing an estimate with an \emph{exponential} approximation guarantee in polynomial time is a highly nontrivial task, achieved by the celebrated LLL algorithm~\cite{lll:lll}; the best current polynomial-time algorithm gives only slightly subexponential approximation factors~\cite{aks:sieve}.  A nearly exponential gap remains between the efficiently achievable approximation ratios and those for which we have complexity-theoretic evidence of hardness.  Also, despite intense effort, no quantum algorithms have been found which improve significantly on their classical counterparts.

A breakthrough work of Ajtai~\cite{ajtai:hardlattice} initiated a sequence of papers~\cite{ajtaidwork:lattice97, regev:lattice04, micciancio:almostperfect, regev:lattice05j, ajtaidwork:lattice07, peikert:svpcrypto} which gave a strong motivation to better understand the approximability of various lattice problems.  These papers build cryptosystems which, remarkably, possess strong \textit{average-case} security, based only on the assumption that certain lattice problems are hard to approximate within polynomial factors in the \emph{worst case}.\footnote{Intriguingly for us, Regev~\cite{regev:qclattice} gave one such cryptosystem based on a \emph{quantum} hardness assumption, and recently Peikert~\cite{peikert:svpcrypto} built upon ideas in~\cite{regev:qclattice} to give a new system based on hardness assumptions against classical algorithms. This is another example of a classical result based on an earlier quantum result. However, the connection is less tight than for the proof systems of Aharonov and Regev, since Peikert replaced the quantum aspect of Regev's earlier work by a fairly unrelated classical approach.}  For a fuller discussion of lattice-based cryptography, see~\cite{micciancio&goldwasser:mgbook,micciancioregev:chapter}.

While these hardness assumptions remain plausible, another sequence of
papers~\cite{lls:kzbases, banasz:transference,
  goldreich&goldwasser:latticeprobs, aharonov&regev:latticeqnp,
  aharonov&regev:latticenpconp} has given evidence that the shortest
vector problem is \emph{not} $\NP$-hard to approximate within polynomial
factors (note that a problem may be intractable without being $\NP$-hard).
The most recent, due to Aharonov and
Regev~\cite{aharonov&regev:latticenpconp}, is a proof system for
$\GapCVP_{c\sqrt{n}}$ for sufficiently large $c > 0$.  That is, if we
are given a basis for a lattice $L$ with the promise that either
$d(v, L) \leq 1$ or $d(v, L) \geq c\sqrt{n}$, there is a proof system
to convince us which of the two cases holds.  One direction is simple:
if $d(v, L) \leq 1$, then exhibiting a lattice vector within distance
1 of $v$ proves this (the promise is not needed in this case).  The
new contribution of~\cite{aharonov&regev:latticenpconp} is to supply
proofs for the case where $d(v, L) \geq c\sqrt{n}$, showing that
$\GapCVP_{c\sqrt{n}}$ is in $\NP \cap \coNP$ (technically, the promise-problem analogue of this class).  It follows that $\GapCVP_{c\sqrt{n}}$ cannot be $\NP$-hard unless the Polynomial Hierarchy collapses; see~\cite{aharonov&regev:latticenpconp} for details of this implication.  

The proof system of~\cite{aharonov&regev:latticenpconp} is, in Aharonov and Regev's own description, a ``dequantization'' of an earlier, \emph{quantum} Merlin-Arthur (\QMA) proof system by the same authors~\cite{aharonov&regev:latticeqnp} for $\GapCVP_{c\sqrt{n}}$.\footnote{Actually, for the quantum protocol of~\cite{aharonov&regev:latticeqnp}, we need the slightly stronger promise that when the input satisfies $d(v, L) \geq c\sqrt{n}$, it also holds that $\lambda(L) \geq c\sqrt{n}$.  A proof system for this restricted problem still yields a proof system for SVP with approximation factor $c\sqrt{n}$.}  
In a $\QMA$ protocol, a computationally unbounded but untrustworthy prover (Merlin) sends a quantum state to a polynomially-bounded quantum verifier (Arthur), who then decides to accept or reject without further communication; unlike $\NP$, $\QMA$ allows a small probability of error. Although Aharonov and Regev analyze their second proof system in a completely classical way, 
the two papers offer an interesting case study of how quantum ideas can be adapted to the classical setting.  We now describe some of the ideas involved; we begin with the quantum proof system of~\cite{aharonov&regev:latticeqnp}, then discuss the classical protocol of~\cite{aharonov&regev:latticenpconp}.

\paragraph{The quantum proof system}
In the $\QMA$ protocol, Prover wants to convince Verifier that $d(v, L) \geq c\sqrt{n}$ for some large $c > 0$.  Central to the proof system is the following geometric idea: for our given lattice $L$, one defines a function $F(x): \mathbb{R}^n \rightarrow [0, \infty)$ that is \emph{lattice-periodic} (i.e., $F(x + y) = F(x)$ for all $x \in \mathbb{R}^n$ and $y \in L$) and is heavily concentrated in balls of radius $\ll \sqrt{n}$ around the lattice points.  Now for each $z \in \mathbb{R}^n$ we consider the ``$z$-shifted'' version of $F$, $F_z(x) = F(x + z)$.  The central idea is that if $d(z, L) \leq 1$, then (informally speaking) $F_z$ has large ``overlap'' with $F$, since the centers of mass shift only slightly.  On the other hand, if $d(z,  L) \geq c\sqrt{n}$, then $F_z$ and $F$ have negligible overlap, since the masses of the two functions are then concentrated on disjoint sets of balls.  In the proof system, Prover aims to convince Verifier that this overlap is indeed negligible when $z = v$ is the target vector.  To this end, Verifier asks to receive the ``state''
$$
\ket{\xi} = \sum_{x \in \mathbb{R}^n} F(x) \ket{x}
$$
describing the pointwise behavior of $F(x)$.  (Here and throughout, we give a simplified, idealized description; the actual protocol uses points $x$ of finite precision, over a bounded region that captures the behavior of the $L$-periodic function $F$.  Thus all sums become finite, and states can be properly normalized.)
We think of $\ket{\xi}$ as the ``correct'' proof which Verifier hopes to receive from Prover; note that $\ket{\xi}$ is independent of $v$.

Verifier cannot hope to recover an arbitrary value from among the exponentially many values stored in $\ket{\xi}$.  However, given $\ket{\xi}$, an elegant technique allows Verifier to estimate the overlap of $F$ with $F_v$, if the overlap $\langle F, F_v\rangle$ is defined (imprecisely) as 
$$
\langle F, F_v\rangle = \sum_{x \in \mathbb{R}^n}F(x) F(x + v ) \,.
$$  
Aharonov and Regev show that the overlap measure $\langle F, F_v\rangle$ they define is extremely close to $e^{-\pi d(v, L)^2/4}$ for any $v$ provided that $\lambda(L) \geq c\sqrt{n}$, so that it is indeed a good indicator of distance from the lattice.  %

To estimate this overlap, Verifier first appends a ``control'' qubit
to $\ket{\xi}$, initialized to $({1}/{\sqrt{2}})\,(\ket{0} + \ket{1})$.  Conditioned on the control qubit being 1, he applies the transformation $T_v$ which takes $\ket{y}$ to $\ket{y - v}$; this yields the state $({1}/{\sqrt{2}})\,(\ket{0}\ket{\xi}  + \ket{1}T_v\ket{\xi})$.  He then applies a Hadamard transformation to the control qubit, yielding 
$$
\frac{1}{2}\bigl[ \ket{0}(\ket{\xi}   + T_v \ket{\xi}  )   +   \ket{1}(\ket{\xi}   - T_v \ket{\xi}  )  \bigr]\,.
$$  
Finally, he measures the control qubit, which equals 1 with probability
$$
\frac{1}{4}(  \bra{\xi}  - \bra{\xi}T_v^{*}  )(\ket{\xi}  - T_v\ket{\xi})     = \frac{1}{2}(1  - {\rm Re}(\bra{\xi}T_v\ket{\xi}))\,.
$$
Consulting the definition of $\ket{\xi}$ and using that $F$ is real-valued, we have ${\rm Re}(\bra{\xi}T_v\ket{\xi}) = \langle F, F_v\rangle$, so the probability of measuring a 1 is linearly related to the overlap Verifier wants to estimate.

The procedure above allows Verifier to estimate $d(v, L)$, \emph{on the assumption} that Prover supplies the correct quantum encoding of $F$.  But Prover could send any quantum state $\ket{\psi}$, so Verifier needs to test that $\ket{\psi}$ behaves something like the desired state $\ket{\xi}$.  In particular, for randomly chosen vectors $z$ within a ball of radius $1/\poly(n)$ around the origin, the overlap $\langle F, F_z\rangle$, estimated by the above procedure, should be about $e^{-\pi \norm{z}^2/4}$.

Consider $h_{\psi}(z) = {\rm Re}(\bra{\psi}T_z\ket{\psi})$ as a function of $z$.  If this function could be arbitrary as we range over choices of $\ket{\psi}$, then observing that $h_{\psi}(z) \approx h_{\xi}(z)$ for short vectors $z$ would give us no confidence that $h_{\psi}(v) \approx \langle F, F_v\rangle$.  However, Aharonov and Regev show that for every $\ket{\psi}$, the function $h_{\psi}(z)$ obeys a powerful constraint called \emph{positive definiteness}.  They then prove that any positive definite function which is sufficiently ``close on average'' to the Gaussian $e^{-\pi \norm{z}^2/4}$ in a small ball around 0, cannot simultaneously be too close to zero at any point within a distance 1 of $L$.  Thus if $d(v, L) \leq 1$, any state Prover sends must either fail to give the expected behavior around the origin, or fail to suggest negligible overlap of $F$ with $F_v$.
We have sketched the core ideas of the quantum proof system; the actual protocol requires multiple copies of $\ket{\xi}$ to be sent, with additional, delicate layers of verification.

\paragraph{The classical proof system}
We now turn to see how some of the ideas of this proof system were adapted in~\cite{aharonov&regev:latticenpconp} to give a classical, deterministic proof system for $\GapCVP_{c\sqrt{n}}$.  The first new insight is that for any lattice $L$, there is an $L$-periodic function $f(x)$ which behaves similarly to $F(x)$ in the previous protocol, and which can be efficiently and accurately approximated at any point, using a polynomial-sized \emph{classical} advice string.  (The function $f(x)$ is defined as a sum of Gaussians, one centered around each point in $L$.) This advice is not efficiently constructible given a basis for $L$, but we may still hope to verify this advice if it is supplied by a powerful prover.

These approximations are derived using ideas of Fourier analysis.  First we introduce the \emph{dual lattice} $L^*$, which consists of all $w \in \mathbb{R}^n$ satisfying $\langle w, y \rangle \in \mathbb{Z}$ for all $y \in L$ (here we use $\langle \cdot, \cdot \rangle$ to denote inner product).  For example, we have $(\mathbb{Z}^n)^* = \mathbb{Z}^n$.  Note that for $w \in L^*$, $r_w(x) = \cos(2\pi  \langle w, x \rangle)$ is an $L$-periodic function.  In fact, any sufficiently smooth $L$-periodic real-valued function, such as $f(x)$, can be uniquely expressed as
$$
f(x) = \sum_{w \in L^*} \widehat{f}(w) \cos(2\pi  \langle w , x\rangle)\,,
$$
where the numbers $\widehat{f}(w) \in \mathbb{R}$ are called the
\emph{Fourier coefficients} of $f$.  Our choice of $f$ has
particularly nicely-behaved Fourier coefficients: they are nonnegative
and sum to 1, yielding a probability distribution (denoted by
$\widehat{f}$).  Thus 
$$
f(x) = \Exp_{w \sim \widehat{f}}[ \cos(2\pi
\langle w , x\rangle)]\,.
$$  It follows from standard sampling ideas that a large enough ($N=\poly(n)$-sized) set of samples $W = w_1, \ldots, w_N$ from $\widehat{f}$ will, with high probability, allow us to accurately estimate $f(x)$ at every point in a large, fine grid about the origin, by the approximation rule
$$
f_W(x) = \frac{1}{N}\sum_{i=1}^N \cos(2\pi   \langle w_i , x\rangle)\,. 
$$
Since $f$ is $L$-periodic and smooth, this yields good estimates everywhere.  So, we let the proof string consist of a matrix $W$ of column vectors $w_1, \ldots, w_N$.  Since Prover claims that the target vector $v$ is quite far from the lattice, and $f$ is concentrated around the lattice points, Verifier expects to see that $f_W(x)$ is small, say, less than $1/2$.

As usual, Verifier must work to prevent Prover from sending a misleading proof.  The obvious first check is that the vectors $w_i$ sent are indeed in $L^*$.  For the next test, a useful fact (which relies on properties of the Gaussians used to define $f$) is that samples drawn from $\widehat{f}$ are not too large and ``tend not to point in any particular direction'': $\Exp_{w \sim \widehat{f}}[ \langle u, w  \rangle^2] \leq {1}/{2\pi}$ for any unit vector~$u$.  This motivates a second test to be performed on the samples $w_1, \ldots, w_N$: check that the largest eigenvalue of $W W^T$ is at most $3N$.

In a sense, this latter test checks that $f_W$ has the correct ``shape'' in a neighborhood of the origin, and plays a similar role to the random-sampling test from the quantum protocol.\footnote{Indeed, the testing ideas are similar, due to the fact that for any $W$, the function $f_W$ obeys the same ``positive definiteness'' property used to analyze the quantum protocol.}  Like its quantum counterpart, this new test is surprisingly powerful: if $W$ is \emph{any} collection of dual-lattice vectors satisfying the eigenvalue constraint, and $d(v, L) \leq 1/100$, then $f_W(v) \geq 1/2$ and Prover cannot use $W$ to make $v$ pass the test.  On the other hand, if $d(v, L) \geq c\sqrt{n}$ for large enough $c$, then choosing the columns of $W$ according to $\widehat{f}$ yields a witness that with high probability passes the two tests and satisfies $f_W(v) < 1/2$.  Scaling by a factor 100 gives a proof system for $\GapCVP_{100c\sqrt{n}}$.

\subsection{A guide to further literature}\label{ssecfurtherlit}
In this section we give pointers to other examples where quantum techniques are used to obtain classical results in some way or other:
\begin{itemize}
\item
{\bf Data structure lower bounds:}
Using linear-algebraic techniques, Radhakrishnan et al.~\cite{rsv:qsetmembershipj} proved lower bounds
on the bit-length of data structures for the set membership problem with quantum decoding algorithms.
Their bounds of course also apply to classical decoding algorithms, but are in fact stronger than the previous 
classical lower bounds of Buhrman et al.~\cite{bmrv:bitvectorsj}.
Sen and Venkatesh did the same for data structures for the predecessor problem~\cite{sen&venkatesh:qcellprobej},
proving a ``round elimination'' lemma in the context of quantum communication complexity that is stronger than the best known classical round elimination lemmas.
A further strengthening of their lemma was subsequently used by Chakrabarti and Regev~\cite{chakrabarti&regev:nn} 
to obtain optimal lower bounds for the approximate nearest neighbour problem.
\item
{\bf Formula lower bounds:}
Recall that a formula is a binary tree whose internal nodes are AND and OR-gates, and each leaf is a Boolean input variable $x_i$ or its negation.
The root of the tree computes a Boolean function of the inputs in the obvious way.
Proving superpolynomial formula lower bounds for specific explicit functions in $\NP$ is a long-standing open problem in complexity theory, 
the hope being that such a result would be a stepping-stone towards the superpolynomial circuit lower bounds that would separate $\Polytime$ from $\NP$
(currently, not even superlinear bounds are known).
The best proven formula-size lower bounds are nearly cubic~\cite{hastad:demorgan}, but a large class of quadratic lower bounds can be obtained
from a quantum result: Laplante et al.~\cite{lls:advformj} showed that the formula size of $f$ is lower bounded by the square of 
the quantum adversary bound for $f$ (mentioned in Section~\ref{secadversary}). Since random functions, as well as many specific functions like Parity and Majority,
have linear adversary bounds, one obtains many quadratic formula lower bounds this way.
\item
{\bf Circuit lower bounds:}
Kerenidis~\cite{kerenidis:qcircuit} describes an approach to prove lower bounds for classical circuit \emph{depth}
using quantum multiparty communication complexity. Roughly speaking, the idea is to combine two classical parties 
into one quantum party, prove lower bounds in the resulting quantum model, and then translate these back to strong lower bounds for the classical model.
(Unfortunately, it seems hard to prove good lower bounds for the resulting quantum model~\cite{lss:quantumnof}.)
\item
{\bf Horn's problem} is to characterize the triples $\mu,\nu$, and
$\lambda$ of vectors of integers for which there exist Hermitian operators
$A$ and $B$ such that $\mu,\nu$, and $\lambda$ are the spectra of $A$,
$B$, and $A+B$, respectively. It is known (and quite non-trivial) that
this problem is equivalent to a question about the representation theory
of the group $\textsf{GL}(d)$ of invertible $d\times d$ complex matrices.
Christandl~\cite{christandl:horn} reproved a slightly weaker version of
this equivalence based on quantum information theory.
\item
{\bf Secret-key distillation:}
In cryptography, Gisin, Renner, and Wolf~\cite{grw:keyagree}
used an analogy with ``quantum bound entanglement'' to provide
evidence against the conjecture that the ``intrinsic information''
in a random variable shared by Alice, Bob, and eavesdropper Eve
always equals the amount of secret key that Alice and Bob can
extract from this; later this conjecture was indeed
disproved~\cite{renner&wolf:new}, though without using quantum methods.
\item
{\bf Equivalence relations and invariants:}
A function $f$ is said to be a \emph{complete invariant} for the equivalence relation $R$ if $(x,y)\in R \Leftrightarrow f(x)=f(y)$.
It is an interesting question whether every polynomial-time computable equivalence relation
has a polynomial-time computable complete invariant.
Fortnow and Grochow~\cite[Theorem~4.3]{fortnow&grochow:equivalence} show this would imply
that the class UP reduces to Simon's problem~\cite{simon:power} (and hence would be in $\BQP$, which seems unlikely).
\item
{\bf Learning positive semidefinite matrices:}
Tsuda, R\"{a}tsch, and Warmuth~\cite{trw:matrixlearning} study online learning of positive semidefinite matrices
using properties of von Neumann divergence (also known as quantum relative entropy) 
in their analysis in order to measure differences between density matrices.
\item
{\bf Johansson's theorem}
asymptotically equates the expected shape of the semi-standard tableau produced 
by a random word in $k$ letters, with the spectrum of a certain random matrix.
Kuperberg~\cite{kuperberg:johansson} provides a proof of this theorem
by treating the random matrix as a quantum random variable on the space of random words.
In another paper, Kuperberg~\cite{kuperberg:tracial} also proved a central limit theorem for non-commutative 
random variables in a von Neumann algebra with a tracial state, based on quantum ideas.
\end{itemize}

\section{Conclusion}
In this paper we surveyed the growing list of applications of quantum computing
techniques to non-quantum problems, in areas ranging from theoretical computer science to
pure mathematics.  
These proofs build on current research in quantum computing, but do not depend
on whether a large-scale quantum computer will ever be built.
One could even go further, and use mathematical frameworks that go ``beyond quantum'' as a proof-tool.
The use of $\PostBQP$ in Section~\ref{sec:postbqp} is in this vein, since we don't expect postselection to be physically implementable.

We feel that ``thinking quantumly'' can be a source of insight and of charming, surprising proofs:
both proofs of new results and simpler proofs of known results (of course, ``simplicity'' is in the eye of the beholder here).
While the examples in this survey do not constitute a full-fledged proof method yet, our
hope is that both the quantum toolbox and its range of applications will continue to grow.

\subsection*{Acknowledgements}
We thank Richard Lipton for his blog entry on these
techniques~\cite{lipton:qmethodblog}, which was one of the things
that motivated us to write this survey.  We thank Scott Aaronson,
Joshua Brody, Iordanis Kerenidis, Greg Kuperberg, Troy Lee,
Fr\'{e}d\'{e}ric Magniez, Ashwin Nayak, Alexander Sherstov, and
Shengyu Zhang for helpful comments and pointers to the literature.
Finally, many thanks to the anonymous Theory of Computing referees for their very
helpful comments.

\bibliographystyle{plain}
%
%\bibliography{new_qproof}

\appendix

\section{The most general quantum model}\label{appgeneral}

The ingredients introduced in Section~\ref{secquantummodel} are all the quantum mechanics we need for the applications of this survey.
However, a more general formalism exists, which we will explain here with a view to future applications---who knows what these may need!

First we generalize pure states.
In the classical world we often have uncertainty about the state of a system,
which can be expressed by viewing the state as a random variable that has a certain
probability distribution over the set of basis states.  Similarly we can define a
\emph{mixed} quantum state as a probability distribution (or ``mixture'') over pure states.
While pure states are written as vectors, it is most convenient to write mixed states as
\emph{density matrices}. A pure state $\ket{\phi}$ corresponds to the rank-one density matrix $\ketbra{\phi}{\phi}$, 
which is the outer product of the vector $\ket{\phi}$ with itself.
A mixed state that is in (not necessarily orthogonal) pure states $\ket{\phi_1},\ldots,\ket{\phi_k}$ with probabilities $p_1,\ldots,p_k$,
respectively, corresponds to the density matrix $\sum_{i=1}^k p_i\ketbra{\phi_i}{\phi_i}$.
The set of density matrices is exactly the set of positive semidefinite (PSD) matrices of trace~1.
A mixed state is pure if, and only if, its density matrix has rank~1.

The most general quantum operation on density matrices is a \emph{completely-positive, trace-preserving (CPTP) map}.
By definition, this is a linear map that sends density matrices to density matrices, even when tensored with the identity operator on another space.
Alternatively, a map ${\cal S}:\rho\mapsto S(\rho)$ from $d\times d$-matrices to $d'\times d'$-matrices is a CPTP map 
if, and only if, it has a \emph{Kraus-representation}:
there are $d'\times d$ matrices $M_1,\ldots,M_k$, satisfying $\sum_{i=1}^k M_i^*M_i=I$, such that ${\cal S}(\rho)=\sum_{i=1}^k M_i\rho M_i^*$ 
for every $d\times d$ density matrix $\rho$.
A unitary map $U$ corresponds to $k=1$ and $M_1=U$, so unitaries act on mixed states by conjugation: $\rho\mapsto U\rho U^*$.
Note that a CPTP map can change the dimension of the state.  For instance, the map that traces out (``throws away'') the second register 
of a 2-register state is a CPTP map. Formally, this map is defined on tensor-product states as $\Tr_2(A\otimes B)=A$, 
and extended to all 2-register states by linearity.

CPTP maps also include measurements as a special case.  For instance, a projective measurement with projectors $P_1,\ldots,P_k$ 
that writes the classical outcome in a second register, corresponds to a CPTP map $\cal S$ with Kraus operators $M_i=P_i\otimes\ket{i}$.
We now have 
$$
{\cal S}(\rho)=\sum_{i=1}^k M_i\rho M_i^*=\sum_{i=1}^k  \frac{P_i\rho P_i^*}{\Tr(P_i\rho P_i^*)}\otimes \Tr(P_i\rho P_i^*)\ketbra{i}{i}\,.
$$
This writes the classical value $i$ in the second register with probability $\Tr(P_i\rho P_i^*)$, and ``collapses'' $\rho$ 
in the first register to its normalized projection in the subspace corresponding to $P_i$.

While this framework of mixed states and CPTP maps looks more
general than the earlier framework of pure states and unitaries,
philosophically speaking it is not: every CPTP map can be
implemented unitarily on a larger space.  What this means is that
for every CPTP map $\cal S$, there exists a state $\rho_0$ on an
auxiliary space, and a unitary on the joint space, such that for
every $\rho$, the state ${\cal S}(\rho)$ equals what one gets by
tracing out the auxiliary register from the state
$U(\rho\otimes\rho_0)U^*$.  We refer to the book of Nielsen and
Chuang~\cite[Section~8.2]{nielsen&chuang:qc} for more details.

\end{document}